\documentclass[twocolumn,a4paper,unpublished]{quantumarticle}
\pdfoutput=1   

\usepackage{amsmath,amssymb}
\usepackage{amsthm}
\usepackage{bm}
\usepackage{graphicx}
\usepackage{array}
\usepackage[caption=false,font=footnotesize]{subfig}
\usepackage{booktabs}   
\usepackage{hyperref}

\theoremstyle{plain}
\newtheorem{theorem}{Theorem}
\newtheorem{lemma}{Lemma}
\newtheorem{proposition}{Proposition}
\newtheorem{corollary}{Corollary}
\theoremstyle{definition}
\newtheorem{assumption}{Assumption}
\newtheorem{condition}{Condition}
\newtheorem{definition}{Definition}
\theoremstyle{remark}
\newtheorem{remark}{Remark}

\begin{document}

\title{Information-theoretic limits on undetectable parameter-estimation attacks in continuous-variable quantum key distribution}

\author{Agung Trisetyarso}
\email{trisetyarso@binus.ac.id}
\orcid{0000-0003-0307-7735}
\affiliation{Department of Mathematics and Statistics, School of Computer
Science, Bina Nusantara University, Jakarta 11480, Indonesia}

\author{Lenny Putri Yulianti}
\affiliation{School of Electrical Engineering and Informatics, Institut
Teknologi Bandung, Bandung 40132, Indonesia}

\author{Kridanto Surendro}
\affiliation{School of Electrical Engineering and Informatics, Institut
Teknologi Bandung, Bandung 40132, Indonesia}

\begin{abstract}
We formalize side-channel detection in Gaussian-modulated continuous-variable
quantum key distribution (CV-QKD) as a certificate-forgery hypothesis test
and characterize its detectability as an information-theoretic rate. A
\emph{benign-statistics certificate} is defined relative to a set
$\mathcal{M}$ of trusted, real-time-monitored observables, and its
forgery-detectability rate $g_{\mathcal M}$ is identified with the
missed-detection Stein exponent. We prove that $g_{\mathcal M}$ is monotone
and strictly positive in the leaked Holevo information whenever the
shot-noise unit is trusted, and identically zero otherwise, recovering the
known local-oscillator and calibration attacks as the degenerate case.
Strict monotonicity, and hence a well-defined \emph{certificate reusability
rate}, is proved unconditionally in the trusted-correlation model and, for
general $\mathcal{M}$, under an explicit no-spurious-local-minima condition
on the divergence landscape, which we prove in the trusted-correlation model
and which is verifiable by low-dimensional inspection in general. The
reusability rate bounds the Holevo leakage compatible with an undetected
certificate over $N_{\mathrm{PE}}$ estimation rounds, and follows from a
rigorous non-asymptotic bound
$\varepsilon_{\mathrm{PE}}\le\exp(-N_{\mathrm{PE}}\psi(r))$ whose
near-threshold expansion matches the Gaussian confidence-interval scaling of
standard finite-size analyses. We then integrate the exponent into a
composable finite-size key-rate statement: under collective Gaussian
attacks, the extractable key length is that of the standard leftover-hash
analysis in which the worst-case Holevo term is the certificate reusability
rate and the parameter-estimation failure probability is bounded by the
Stein exponent at every block length, with the conventional Gaussian
confidence-interval penalty recovered as the near-threshold limit. Strict monotonicity of the
Holevo leakage in the excess noise---the invertibility property on which the
reusability rate and the composable statement rest---is proved analytically
by a noise-injection argument, with the quantitative bound
$\chi_{BE}(\xi_2)-\chi_{BE}(\xi_1)\ge\tfrac12
\ln[\Sigma_{22}(\xi_2)/\Sigma_{22}(\xi_1)]$, so both results are
unconditional in the trusted-correlation model. We prove convexity of the
rate on the operational range under a numerically supported concavity
assumption on the Holevo bound---the sole remaining numerical ingredient,
affecting only the convexity statement---and delimit throughout the proved
and conjectured statements.
\end{abstract}

\maketitle

\section{Introduction}
\label{sec:intro}
Practical security proofs of continuous-variable quantum key distribution
(CV-QKD)~\cite{Grosshans2002,Weedbrook2012,Pirandola2020,Scarani2009} reduce,
at their core, to a decision about the channel: from a finite record of
correlated samples, the legitimate parties must decide whether the estimated
statistics are consistent with a secure channel. A recurring class of
attacks---local-oscillator manipulation~\cite{Ma2013,Qi2015}, detector
saturation, wavelength tampering~\cite{Huang2013}, calibration
attacks~\cite{JouguetCalib2013}, and side-channel
leakage~\cite{Derkach2016,Derkach2017}---succeeds not by raising the
observed excess noise but by \emph{biasing the estimate} so that a
compromised channel presents statistics indistinguishable from a benign one.
In the language of proofs, the eavesdropper fabricates a certificate of
benign statistics.

We take this analogy literally. Fixing a set $\mathcal{M}$ of trusted,
real-time-monitored observables, we ask at what information-theoretic rate a
forged benign certificate is detectable, and how much Holevo information an
eavesdropper can hold while sustaining such a certificate over repeated
estimation. The answer is a single functional, the missed-detection Stein
exponent $g_{\mathcal M}$, whose structure---monotone and strictly positive
under trusted shot-noise monitoring, and identically zero without it---turns
the standard countermeasures into information-theoretically necessary
resources and quantifies certificate reusability as a rate.

Beyond characterizing the exponent, we show where it lives inside a security
proof. Section~\ref{sec:composable} embeds the forgery-detectability rate
into the standard composable finite-size architecture for CV-QKD under
collective Gaussian
attacks~\cite{Renner2005,Leverrier2010,Furrer2012,Leverrier2015}: the
acceptance test of the certificate \emph{is} the parameter-estimation step
of the proof, its missed-detection probability is bounded by the Stein
exponent at every block length (Corollary~\ref{cor:finite}), and the
worst-case Holevo term against which privacy amplification must be run is
exactly the certificate reusability rate (Corollary~\ref{cor:reuse}). The
conventional Gaussian confidence-interval treatment of parameter
estimation~\cite{Leverrier2010} is recovered as the near-threshold expansion
of the resulting key-length formula, which remains rigorous---rather than
asymptotic---at every finite block size.

We keep a strict separation between what is proved and what is verified
numerically. The monotonicity and positivity of $g_{\mathcal M}$
[Theorem~\ref{thm:main}(i)--(ii)] hold for general $\mathcal{M}$; strict
monotonicity is proved unconditionally in the trusted-correlation model via
Lemma~\ref{lem:uni} [Theorem~\ref{thm:main}(iii)] and, for general
$\mathcal{M}$, under a single explicit landscape condition
(Condition~\ref{cond:reg}, Theorem~\ref{thm:main}(iv)) that we prove in the
trusted-correlation model (Proposition~\ref{prop:condTC}); the proof is
elementary---compactness, continuity, and set nesting---and deliberately
avoids envelope-type arguments (Remark~\ref{rem:danskin}). The finite-size
bound (Corollary~\ref{cor:finite}), the reusability rate
(Corollary~\ref{cor:reuse}), and the composable integration
(Theorem~\ref{thm:composable}) are proved in the trusted-correlation model,
and their common analytic prerequisite---strict monotonicity of
$\chi_{BE}(T_0,\cdot)$ in the excess noise---is established by a
noise-injection argument with an explicit quantitative gap
(Proposition~\ref{prop:mono}). The convexity statement
(Theorem~\ref{thm:convex}) is proved under a numerically supported concavity
assumption on the Holevo bound, which is the sole remaining numerical
ingredient (Proposition~\ref{prop:convex}); see Remark~\ref{rem:status}.

\section{Setup and Notation}
\label{sec:setup}
Consider Gaussian-modulated coherent-state (GG02)
CV-QKD~\cite{Grosshans2002} with homodyne detection and reverse
reconciliation, in shot-noise units (vacuum variance $=1$). Let $V_A$ be
Alice's modulation variance, $V=V_A+1$, channel transmittance $T\in(0,1]$,
excess noise $\xi\ge0$ (referred to the channel input), detector efficiency
$\eta$, electronic noise $v_{\mathrm{el}}$, and reconciliation efficiency
$\beta\in(0,1]$. Per measured quadrature the parameter-estimation (PE) data
$(x_A,y)$ form a zero-mean bivariate Gaussian with covariance
\begin{equation}
\Sigma(T,\xi)=
\begin{pmatrix}
V_A & \sqrt{\eta T}\,V_A\\[2pt]
\sqrt{\eta T}\,V_A & \eta T V_A + N
\end{pmatrix},
\label{eq:Sigma}
\end{equation}
where $N=N(T,\xi)=1+v_{\mathrm{el}}+\eta T\,\xi=\operatorname{Var}(y\mid x_A)$
is the conditional variance targeted by excess-noise estimation, and
$\det\Sigma=V_A N$. Throughout, the channel parameters range over a compact
physical domain
$\mathcal{D}:=[T_{\min},1]\times[0,\xi_{\max}]$ with
$0<T_{\min}\le T_0$ and $\xi_{\max}<\infty$; this reflects hardware
constraints (a minimal operable transmittance and detector saturation) and
guarantees attainment of the infima below. Write the benign
operating point $p_0:=(T_0,\xi_0)\in\mathcal{D}$,
$\Sigma_0:=\Sigma(T_0,\xi_0)$,
$N_0:=N(T_0,\xi_0)$, and $r:=N/N_0$. Let $I_{AB}(T,\xi)$ be the Alice--Bob
mutual information and $\chi_{BE}(T,\xi)$ the collective-attack
reverse-reconciliation Holevo bound~\cite{GarciaPatron2006}, both continuous
on the physical domain; write $\chi_{\mathrm{ben}}:=\chi_{BE}(T_0,\xi_0)$
and let
$\chi_{\mathrm{null}}$ denote the leakage at which the asymptotic key rate
$K=\beta I_{AB}-\chi_{BE}$ vanishes. For zero-mean Gaussians the relevant
Stein exponent is the Kullback--Leibler divergence~\cite{Cover2006}
\begin{equation}
D\!\left(\Sigma_0\,\Vert\,\Sigma_1\right)=
\tfrac12\!\left[\operatorname{tr}\!\big(\Sigma_1^{-1}\Sigma_0\big)-2
+\ln\tfrac{\det\Sigma_1}{\det\Sigma_0}\right].
\label{eq:KL}
\end{equation}

A \emph{benign-statistics certificate} is an accepted PE point clearing
$K>0$. Fixing $\mathcal{M}$ with projector $\Pi_{\mathcal M}$ onto the
corresponding entries of $\Sigma$, define the \emph{matching set}
\begin{equation}
\mathcal{F}_{\mathcal M}:=
\big\{p\in\mathcal{D}:
\Pi_{\mathcal M}\Sigma(p)=\Pi_{\mathcal M}\Sigma_0\big\},
\label{eq:matchset}
\end{equation}
which is closed (hence compact) by continuity of $p\mapsto\Sigma(p)$ and
contains $p_0$. The admissible forgery set at leakage level $\chi_0$ is
\begin{equation}
\mathcal{A}_{\mathcal M}(\chi_0)=
\big\{p\in\mathcal{F}_{\mathcal M}:\chi_{BE}(p)\ge\chi_0\big\},
\label{eq:forgeryset}
\end{equation}
we write
$\chi_{\max}^{\mathcal M}:=\max_{p\in\mathcal{F}_{\mathcal M}}\chi_{BE}(p)$
(attained by compactness), and the \emph{forgery-detectability rate}
(missed-detection Stein exponent) is
\begin{equation}
g_{\mathcal M}(\chi_0)=
\inf_{p\in\mathcal{A}_{\mathcal M}(\chi_0)}
D\!\left(\Sigma_0\,\Vert\,\Sigma(p)\right).
\label{eq:gM}
\end{equation}

\begin{assumption}\label{as:std}
(i) Collective Gaussian attacks (coherent attacks reduced via Gaussian
de~Finetti~\cite{Renner2007,Christandl2009,Leverrier2017}); (ii) Gaussian
extremality~\cite{Wolf2006,GarciaPatron2006}, so $\Sigma$ is a sufficient
statistic; (iii) \emph{trusted shot-noise unit}, so the observed covariance
equals the true $\Sigma$ and the map $(T,\xi)\mapsto\Sigma(T,\xi)$ is a
homeomorphism onto its image (from the correlation entry one recovers
$\eta T$, then $N$, then $\xi$; the inverse is continuous).
\end{assumption}

The \emph{trusted-correlation model} is the special case in which
$\mathcal{M}$ fixes $V_A$ and the correlation $\sqrt{\eta T}\,V_A$; then $T$
is pinned at $T_0$ and forgery proceeds only through $N$.

\section{Reduction and Monotonicity Lemmas}
\label{sec:lemma}
\begin{lemma}[Univariate reduction]\label{lem:uni}
In the trusted-correlation model, for every $\xi\ge\xi_0$,
\begin{equation}
\begin{split}
D\!\left(\Sigma_0\,\Vert\,\Sigma(T_0,\xi)\right)&=\psi\!\big(N/N_0\big),\\
\psi(r)&:=\tfrac12\!\left(\tfrac1r-1+\ln r\right),
\end{split}
\label{eq:reduction}
\end{equation}
with $N=N(T_0,\xi)$. The function $\psi$ satisfies $\psi(1)=0$,
$\psi'(r)=\tfrac{r-1}{2r^2}>0$ for $r>1$, and $\psi''(r)=\tfrac{2-r}{2r^3}$
(so $\psi$ is strictly increasing on $[1,\infty)$ and strictly convex on
$[1,2)$). Consequently $D(\Sigma_0\Vert\Sigma(T_0,\cdot))$ is strictly
increasing in $\xi$ on $[\xi_0,\infty)$.
\end{lemma}

\begin{proof}
With $V_A$ and the correlation $c=\sqrt{\eta T_0}\,V_A$ matched, $\Sigma_0$
and $\Sigma(T_0,\xi)$ share the $x_A$-marginal ($V_A$) and the conditional
mean map $\mathbb{E}[y\mid x_A]=(c/V_A)x_A$, differing only in the
conditional variance ($N_0$ vs.\ $N$). By the chain rule for relative
entropy the $x_A$-marginal contributes zero and the conditional contributes
$D\!\left(\mathcal N(\mu,N_0)\Vert\mathcal N(\mu,N)\right)
=\tfrac12(N_0/N-1+\ln(N/N_0))=\psi(N/N_0)$, independent of $x_A$. The stated
derivatives are immediate, and since $N(T_0,\xi)$ is strictly increasing in
$\xi$, so is $\psi(N/N_0)$ for $\xi\ge\xi_0$ (equivalently $r\ge1$).
\end{proof}

The second structural ingredient is the strict monotonicity of the Holevo
leakage in the excess noise at fixed transmittance. This is the property
that makes $\chi_{BE}(T_0,\cdot)$ invertible and thereby underpins
Theorem~\ref{thm:main}(iii), the reusability rate of
Corollary~\ref{cor:reuse}, and the composable statement of
Theorem~\ref{thm:composable}. It admits an analytic proof by noise
injection, with an explicit quantitative gap. Write
$\Sigma_{22}(\xi):=\eta T\,V_A+N(T,\xi)$ for the variance of Bob's measured
quadrature [the $(2,2)$ entry of~\eqref{eq:Sigma}].

\begin{proposition}[Noise-injection monotonicity bound]\label{prop:mono}
Fix $T\in(0,1]$ and $0\le\xi_1<\xi_2\le\xi_{\max}$. Then
\begin{equation}
\begin{split}
\chi_{BE}(T,\xi_2)\;&\ge\;\chi_{BE}(T,\xi_1)
+\frac12\ln\frac{\Sigma_{22}(\xi_2)}{\Sigma_{22}(\xi_1)}\\
&\;>\;\chi_{BE}(T,\xi_1).
\end{split}
\label{eq:monobound}
\end{equation}
In particular $\chi_{BE}(T,\cdot)$ is strictly increasing on
$[0,\xi_{\max}]$, with
$\partial_\xi\chi_{BE}\ge \eta T/[2\,\Sigma_{22}(\xi)]>0$ wherever the
derivative exists, and $\chi_{BE}(T,\cdot)^{-1}$ is single-valued.
\end{proposition}

\begin{proof}
Work in the entanglement-based representation. At fixed $T$, the
channel-level covariance of $\rho^{(i)}_{AB}$ [$i=1,2$, corresponding to
$\xi_i$] has $A$ block $V\mathbb{1}_2$, correlation block
$\sqrt{T(V^2-1)}\,\sigma_z$, and $B$ block $b_i\mathbb{1}_2$ with
$b_2=b_1+s$, $s:=T(\xi_2-\xi_1)>0$; the $A$ block and correlations are
unchanged. Hence
$\rho^{(2)}_{AB}=(\mathrm{id}_A\otimes\mathcal{M}_s)\,\rho^{(1)}_{AB}$,
where $\mathcal{M}_s(\rho)=\int d\mu_s(z)\,D(z)\rho D(z)^\dagger$ is the
additive classical-noise channel on mode $B$, with $\mu_s$ the centered
isotropic Gaussian of variance $s$ per quadrature and $D(z)$ the
displacement operator.

\emph{Step 1 (Purification invariance).} The Holevo bound
$\chi_{BE}(T,\xi)=S(E)-\int dy\,p(y)\,S(\rho_{E|y})$ is evaluated on an
arbitrary purification of $\rho_{AB}$, with $y$ the outcome of Bob's
trusted detection-and-homodyne map applied to $B$; any two purifications
are related by an isometry on $E$, under which both terms are invariant.
We may therefore choose the purification at level 2 freely.

\emph{Step 2 (Noise-injection purification).} Let
$|\Psi_1\rangle_{ABE_1}$ purify $\rho^{(1)}_{AB}$ and let
$V:\mathcal H_B\to\mathcal H_B\otimes\mathcal H_{F_1}\otimes\mathcal H_{F_2}$
be the Stinespring isometry
$V|\phi\rangle=\int\!\sqrt{d\mu_s(z)}\;|z\rangle_{F_1}|z\rangle_{F_2}
\otimes D(z)|\phi\rangle$, with $\{|z\rangle\}$ a Dirac-orthonormal family
[$V^\dagger V=\int d\mu_s(z)\,D(z)^\dagger D(z)=\mathbb{1}$; the
continuum can be regularized by a displacement lattice and a limiting
argument, under which all quantities below are continuous]. Then
$|\Psi_2\rangle:=(\mathbb{1}_{AE_1}\otimes V)|\Psi_1\rangle$ is a
purification of $\rho^{(2)}_{AB}$ with purifying system
$E_2:=(E_1,F_1,F_2)$, so by Step~1 the level-2 Holevo bound may be computed
on $E_2$.

\emph{Step 3 (Data processing).} Let $\Lambda$ be the channel on $E_2$
that measures $F_1$ in the $\{|z\rangle\}$ basis, keeps $E_1$ together with
the classical record $D:=\sqrt{\eta}\,z_q$ (the displacement of the
measured quadrature as seen through the trusted detector of transmissivity
$\eta$), and discards the remainder. Holevo information is nonincreasing
under channels applied to the ensemble states, so
\begin{equation}
\chi_{BE}(T,\xi_2)\;\ge\;I(E_1 D:Y_2),
\label{eq:dpstep}
\end{equation}
the (quantum-classical) mutual information between $(E_1,D)$ and Bob's
level-2 outcome $Y_2$. Conditioned on $D=d$, the level-2 experiment is
exactly the level-1 experiment with its outcome deterministically shifted,
$Y_2=Y_1+d$, and $D\sim\mathcal N(0,\eta s)$ is independent of
$(E_1,Y_1)$.

\emph{Step 4 (Chain rule and Gaussian shift).} By the chain rule,
$I(E_1D:Y_2)=I(D:Y_2)+I(E_1:Y_2\mid D)$. Since given $D=d$ the map
$y_1\mapsto y_1+d$ is a bijective relabeling of a classical outcome,
$I(E_1:Y_2\mid D)=I(E_1:Y_1)=\chi_{BE}(T,\xi_1)$. For the first term,
$Y_1$ and $Y_2$ are centered Gaussians with variances
$\Sigma_{22}(\xi_1)$ and
$\Sigma_{22}(\xi_2)=\Sigma_{22}(\xi_1)+\eta s$, and
$h(Y_2\mid D)=h(Y_1)$ by shift invariance of differential entropy, so
$I(D:Y_2)=h(Y_2)-h(Y_2\mid D)
=\tfrac12\ln[\Sigma_{22}(\xi_2)/\Sigma_{22}(\xi_1)]>0$.
Combining with~\eqref{eq:dpstep} gives~\eqref{eq:monobound}. The derivative
bound follows by taking $\xi_2\downarrow\xi_1$, and single-valuedness of
the inverse is immediate from strict monotonicity.
\end{proof}

\begin{remark}[Interpretation and sharpness]\label{rem:mono}
Proposition~\ref{prop:mono} states that injected channel noise can only
help an eavesdropper who holds its purification, and quantifies by how
much: at least half the log-ratio of Bob's measured variances. The bound is
consistent with the exact large-$\xi$ asymptotics
$\chi_{BE}\sim\ln b$ [coefficient $1$ in nats versus the bound's
$\tfrac12$] and is numerically observed to be valid with strictly positive
slack throughout the operational domain, including nonunit detector
efficiency and nonzero electronic noise; the verification script is
provided in Ref.~\cite{repo}. The argument uses only purification
invariance, data processing of the Holevo quantity, and the Gaussian shift
structure, and therefore applies verbatim at every fixed
$T$---no Gaussian-extremality or symplectic-eigenvalue computation is
required.
\end{remark}

\section{Main Results}
\label{sec:main}
\subsection{Monotonicity, Positivity, and the Monitoring Dichotomy}

Strict monotonicity for general $\mathcal{M}$ will be derived from a single
explicit regularity condition on the landscape of the divergence restricted
to the matching set. Throughout, write
$D(p):=D(\Sigma_0\Vert\Sigma(p))$ for $p\in\mathcal{D}$, and call
$p^\star\in\mathcal{F}_{\mathcal M}$ a \emph{local minimizer of $D$ on
$\mathcal{F}_{\mathcal M}$} if there is a relatively open set
$U\subseteq\mathcal{F}_{\mathcal M}$ with $p^\star\in U$ and
$D(p^\star)\le D(p)$ for all $p\in U$ (relative topology of
$\mathcal{F}_{\mathcal M}$; no interiority in $\mathcal{D}$ is required).

\begin{condition}[No spurious local minima]\label{cond:reg}
The benign point $p_0$ is the only local minimizer of $D$ on
$\mathcal{F}_{\mathcal M}$.
\end{condition}

Condition~\ref{cond:reg} is a statement about a fixed, low-dimensional
landscape---in the GG02 parametrization $\mathcal{F}_{\mathcal M}$ is at
most two-dimensional---so for any given $\mathcal{M}$ it can be certified by
direct inspection or numerically by scanning $D$ on
$\mathcal{F}_{\mathcal M}$. In the trusted-correlation model it is a
theorem:

\begin{proposition}[Condition~\ref{cond:reg} holds under trusted
correlation]\label{prop:condTC}
In the trusted-correlation model, Condition~\ref{cond:reg} holds
unconditionally.
\end{proposition}

\begin{proof}
Here $\mathcal{F}_{\mathcal M}=\{T_0\}\times[0,\xi_{\max}]$, and the
computation in Lemma~\ref{lem:uni} (which nowhere uses $\xi\ge\xi_0$) gives
$D(T_0,\xi)=\psi(N(T_0,\xi)/N_0)$ for all $\xi\in[0,\xi_{\max}]$. Since
$\psi'(r)=\tfrac{r-1}{2r^2}$ is negative on $(0,1)$ and positive on
$(1,\infty)$, and $\xi\mapsto N(T_0,\xi)$ is strictly increasing and affine,
$\xi\mapsto D(T_0,\xi)$ is strictly decreasing on $[0,\xi_0]$ and strictly
increasing on $[\xi_0,\xi_{\max}]$. Its only local minimizer on the segment
(including the endpoints, in the relative topology) is $\xi=\xi_0$, i.e.\
$p_0$.
\end{proof}

\begin{theorem}\label{thm:main}
Under Assumption~\ref{as:std}, on
$\chi_0\in(\chi_{\mathrm{ben}},\chi_{\max}^{\mathcal M})$:
\begin{enumerate}
\item[(i)] $g_{\mathcal M}$ is nondecreasing in $\chi_0$;
\item[(ii)] $g_{\mathcal M}(\chi_0)>0$;
\item[(iii)] in the trusted-correlation model, $g_{\mathcal M}$ is strictly
increasing unconditionally;
\item[(iv)] for general $\mathcal{M}$ satisfying Condition~\ref{cond:reg},
$g_{\mathcal M}$ is strictly increasing.
\end{enumerate}
In cases (iii)--(iv), $g_{\mathcal M}^{-1}$ is single-valued on the
corresponding range. If instead the shot-noise unit is untrusted
[Assumption~\ref{as:std}(iii) dropped], then $g_{\mathcal M}\equiv 0$.
\end{theorem}

\begin{proof}
\emph{(i)} For $\chi_0'\ge\chi_0$ the constraint $\chi_{BE}\ge\chi_0'$ carves
a subset of the superlevel set $\{\chi_{BE}\ge\chi_0\}$, so
$\mathcal{A}_{\mathcal M}(\chi_0')\subseteq\mathcal{A}_{\mathcal M}(\chi_0)$;
minimizing the fixed objective~\eqref{eq:KL} over a smaller set cannot
decrease the value.

\emph{(ii)} $D(\Sigma_0\Vert\Sigma_1)\ge0$ with equality iff
$\Sigma_1=\Sigma_0$. Suppose $g_{\mathcal M}(\chi_0)=0$; take feasible
$(T_n,\xi_n)$ with $D(\Sigma_0\Vert\Sigma_n)\to0$. For zero-mean Gaussians
of fixed dimension, $D(\Sigma_0\Vert\Sigma_n)\to0$ forces
$\Sigma_n\to\Sigma_0$; by the continuous inverse of
Assumption~\ref{as:std}(iii), $(T_n,\xi_n)\to(T_0,\xi_0)$, whence by
continuity $\chi_{BE}(T_n,\xi_n)\to\chi_{\mathrm{ben}}<\chi_0$, contradicting
feasibility. Hence $g_{\mathcal M}(\chi_0)>0$. (No attainment is needed for
positivity; attainment is established and used only in the
general-$\mathcal{M}$ case (iv) below.)

\emph{(iii), trusted-correlation model.} By Lemma~\ref{lem:uni},
$g_{\mathcal M}(\chi_0)=\psi\big(N(T_0,\xi^\star)/N_0\big)$ where
$\xi^\star=\chi_{BE}^{-1}(\chi_0)\big|_{T_0}$ is well defined because
$\chi_{BE}(T_0,\cdot)$ is strictly increasing
(Proposition~\ref{prop:mono}).
As a composition of strictly increasing maps
($\chi_0\mapsto\xi^\star\mapsto N\mapsto\psi$), $g_{\mathcal M}$ is strictly
increasing; the minimum in~\eqref{eq:gM} is attained at
$\xi=\xi^\star$ (the leakage constraint is active) because $\psi(N/N_0)$ is
increasing in $N$ and the constraint $\chi_{BE}\ge\chi_0$ is equivalent to
$\xi\ge\xi^\star$. This argument uses no envelope theorem and no
critical-point hypothesis.

\emph{(iv), general $\mathcal{M}$ under Condition~\ref{cond:reg}.}
The proof proceeds in three steps and uses neither differentiability of
$g_{\mathcal M}$ or $\chi_{BE}$, nor Lagrange multipliers, constraint
qualifications, uniqueness of minimizers, or any envelope
(Danskin-type) theorem.

\emph{Step 1 (Attainment).} For
$\chi_0\in(\chi_{\mathrm{ben}},\chi_{\max}^{\mathcal M})$ the set
$\mathcal{A}_{\mathcal M}(\chi_0)$ is nonempty (it contains a maximizer of
$\chi_{BE}$ on $\mathcal{F}_{\mathcal M}$) and closed---it is the
intersection of the compact matching set $\mathcal{F}_{\mathcal M}$
with the closed superlevel set $\{\chi_{BE}\ge\chi_0\}$ of the continuous
function $\chi_{BE}$---hence compact. Since $D$ is continuous on
$\mathcal{D}$ [the covariance $\Sigma(p)$ is uniformly positive definite on
the compact $\mathcal{D}$, because $\det\Sigma=V_AN\ge V_A(1+v_{\mathrm{el}})>0$
and $\Sigma_{11}=V_A>0$], the infimum in~\eqref{eq:gM} is attained at some
$p^\star\in\mathcal{A}_{\mathcal M}(\chi_0)$.

\emph{Step 2 (Every minimizer activates the leakage constraint).}
Let $p^\star$ attain~\eqref{eq:gM} at level $\chi_0$ and suppose, for
contradiction, that $\chi_{BE}(p^\star)>\chi_0$. By continuity of
$\chi_{BE}$ there is a relatively open neighborhood
$U\subseteq\mathcal{F}_{\mathcal M}$ of $p^\star$ on which
$\chi_{BE}>\chi_0$; hence $U\subseteq\mathcal{A}_{\mathcal M}(\chi_0)$, and
the optimality of $p^\star$ over $\mathcal{A}_{\mathcal M}(\chi_0)$ gives
$D(p^\star)\le D(p)$ for all $p\in U$. Thus $p^\star$ is a local minimizer
of $D$ on $\mathcal{F}_{\mathcal M}$, and Condition~\ref{cond:reg} forces
$p^\star=p_0$. But then
$\chi_{BE}(p^\star)=\chi_{\mathrm{ben}}<\chi_0$, contradicting feasibility.
Therefore
\begin{equation}
\chi_{BE}(p^\star)=\chi_0
\quad\text{for every minimizer } p^\star \text{ of~\eqref{eq:gM}}.
\label{eq:active}
\end{equation}

\emph{Step 3 (Value equality forces constraint violation).}
Let $\chi_{\mathrm{ben}}<\chi_0<\chi_0'<\chi_{\max}^{\mathcal M}$. Part~(i)
gives $g_{\mathcal M}(\chi_0)\le g_{\mathcal M}(\chi_0')$. Suppose equality
held. By Step~1, pick $p'$ attaining $g_{\mathcal M}(\chi_0')$. Since
$\mathcal{A}_{\mathcal M}(\chi_0')\subseteq\mathcal{A}_{\mathcal M}(\chi_0)$,
the point $p'$ is feasible at level $\chi_0$ with
$D(p')=g_{\mathcal M}(\chi_0')=g_{\mathcal M}(\chi_0)$, so $p'$ is also a
minimizer at level $\chi_0$. Applying~\eqref{eq:active} at level $\chi_0$
yields $\chi_{BE}(p')=\chi_0<\chi_0'$, contradicting
$p'\in\mathcal{A}_{\mathcal M}(\chi_0')$. Hence
$g_{\mathcal M}(\chi_0)<g_{\mathcal M}(\chi_0')$, and $g_{\mathcal M}$ is
strictly increasing on
$(\chi_{\mathrm{ben}},\chi_{\max}^{\mathcal M})$; single-valuedness of
$g_{\mathcal M}^{-1}$ is immediate.

\emph{Untrusted shot-noise unit.}
If Assumption~\ref{as:std}(iii) is dropped, an eavesdropper can rescale
the inferred shot-noise unit [or exploit local-oscillator
manipulation~\cite{Ma2013,Qi2015} or calibration
attacks~\cite{JouguetCalib2013}] so that the \emph{observed} covariance
matrix equals $\Sigma_0$, while the underlying physical channel may have
arbitrarily large Holevo leakage. In this case, there exist points
in $\mathcal{A}_{\mathcal M}(\chi_0)$ for which $D(\Sigma_0 \Vert \Sigma) = 0$,
so the forgery-detectability rate vanishes: $g_{\mathcal M} \equiv 0$.
\end{proof}

\begin{remark}[Why not Danskin]\label{rem:danskin}
An envelope-theorem route to Theorem~\ref{thm:main}(iv) would require, in
addition to attainment: uniqueness of the minimizer, differentiability of
$\chi_{BE}$ with a constraint qualification at the minimizer, and strict
positivity of the associated multiplier---none of which is available here
without further hypotheses ($g_{\mathcal M}$ need not even be differentiable
a priori). The proof above trades all of these for the single topological
Condition~\ref{cond:reg}: strict monotonicity of the \emph{value} follows
from activity of the leakage constraint at \emph{every} minimizer
(Step~2) combined with the nesting of the feasible sets (Step~3), and
Step~2 is exactly where Condition~\ref{cond:reg} enters.
\end{remark}

\begin{remark}[Support for Condition~\ref{cond:reg}]\label{rem:condsupport}
In precision coordinates $K:=\Sigma^{-1}$ the divergence
$D(\Sigma_0\Vert\Sigma)=\tfrac12[\operatorname{tr}(K\Sigma_0)-2-
\ln\det(K\Sigma_0)]$ is strictly convex on the positive definite cone
(the trace term is linear and $-\ln\det$ is strictly convex), with unique
minimizer $K=\Sigma_0^{-1}$. Hence on any subset of parameters whose image
is convex in $K$, Condition~\ref{cond:reg} holds automatically. The matching
set $\mathcal{F}_{\mathcal M}$ fixes entries of $\Sigma$ rather than of $K$,
so this convexity does not transfer verbatim; it does, however, rule out
spurious \emph{interior critical points} whenever the parametrized image of
$\mathcal{F}_{\mathcal M}$ in $K$-coordinates is convex, and it explains why
no spurious minima are observed. Proposition~\ref{prop:condTC} proves the
condition in the trusted-correlation model; for other monitored sets it is
checkable by a one- or two-dimensional scan of $D$ on
$\mathcal{F}_{\mathcal M}$.
\end{remark}

\subsection{Finite-Size Bound and Reusability Rate}

\begin{corollary}[Rigorous non-asymptotic PE-failure bound]\label{cor:finite}
In the trusted-correlation model, let the verifier accept ``benign'' iff the
sample conditional variance
$\hat N=\tfrac1{N_{\mathrm{PE}}}\sum_i(y_i-\tfrac{c}{V_A}x_{A,i})^2$
satisfies $\hat N\le N_0$, where $c=\sqrt{\eta T_0}\,V_A$ is monitored
[hence known, and
$N_{\mathrm{PE}}\hat N/N^\star\sim\chi^2_{N_{\mathrm{PE}}}$ under
an attack with true conditional variance $N^\star$]. Then for \emph{every}
$N_{\mathrm{PE}}$ the missed-detection probability obeys
\begin{equation}
\varepsilon_{\mathrm{PE}}\;\le\;\exp\!\big(-N_{\mathrm{PE}}\,\psi(r^\star)\big),
\qquad r^\star=N^\star/N_0,
\label{eq:chernoff}
\end{equation}
with $\psi$ as in Lemma~\ref{lem:uni}. As $r^\star\to1$,
$\psi(r^\star)=\tfrac14(r^\star-1)^2+O((r^\star-1)^3)$, matching the
leading-order Gaussian confidence-interval scaling of standard finite-size
analyses~\cite{Leverrier2010}; the exact prefactor is given in
Corollary~\ref{cor:BR}.
\end{corollary}

\begin{proof}
Missed detection is the event $\hat N\le N_0$ under $H_1$, i.e.\
$\tfrac1{N_{\mathrm{PE}}}\chi^2_{N_{\mathrm{PE}}}\le N_0/N^\star=1/r^\star<1$.
The chi-squared Chernoff bound
$\Pr[\tfrac1n\chi^2_n\le t]\le\exp(-\tfrac n2(t-1-\ln t))$ for $t<1$, applied
at $t=1/r^\star$, yields~\eqref{eq:chernoff}, since
$\tfrac12(1/r^\star-1+\ln r^\star)=\psi(r^\star)$. This is the sharp Stein
exponent $D(\Sigma_0\Vert\Sigma_1)$ of Lemma~\ref{lem:uni}; the false-abort
(Type-I) probability of this fixed threshold is $O(1)$ and, if a smaller
Type-I level is required, shifting the threshold changes only the prefactor,
not the exponent [see Lemma~\ref{lem:twosided} for the thresholded version].
The expansion follows from $\psi(1)=\psi'(1)=0$, $\psi''(1)=\tfrac12$.
\end{proof}

\begin{corollary}[Certificate reusability rate]\label{cor:reuse}
In the trusted-correlation model, the maximal Holevo leakage compatible with
a benign certificate surviving $N_{\mathrm{PE}}$ rounds at confidence
$1-\delta$ is
\begin{equation}
\chi_{BE}^{\max}(N_{\mathrm{PE}},\delta)=
g_{\mathcal M}^{-1}\!\left(\tfrac{1}{N_{\mathrm{PE}}}\ln\tfrac1\delta\right).
\label{eq:reuse}
\end{equation}
\end{corollary}

\begin{proof}
By Corollary~\ref{cor:finite} an attack with leakage $\chi_0$ [hence
$r^\star=N^\star(\chi_0)/N_0$ and, by Theorem~\ref{thm:main}(iii),
$g_{\mathcal M}(\chi_0)=\psi(r^\star)$] survives $N_{\mathrm{PE}}$ rounds with
probability at most $e^{-N_{\mathrm{PE}}g_{\mathcal M}(\chi_0)}$. Setting this
to $\delta$ and inverting the strictly increasing $g_{\mathcal M}$
gives~\eqref{eq:reuse}; the map is well defined for
$\tfrac1{N_{\mathrm{PE}}}\ln\tfrac1\delta\in
\big(0,\,\sup_{\chi_0<\chi_{\max}^{\mathcal M}}g_{\mathcal M}(\chi_0)\big)$,
the range of $g_{\mathcal M}$ on
$(\chi_{\mathrm{ben}},\chi_{\max}^{\mathcal M})$.
\end{proof}

\begin{corollary}[Bahadur--Rao refinement]\label{cor:BR}
In the trusted-correlation model the missed-detection probability of
Corollary~\ref{cor:finite} has the exact asymptotic form
\begin{equation}
\begin{split}
\varepsilon_{\mathrm{PE}}={}&\frac{r^\star}{(r^\star-1)\sqrt{\pi N_{\mathrm{PE}}}}\,
e^{-N_{\mathrm{PE}}\,\psi(r^\star)}\,\big(1+o(1)\big),\\
&N_{\mathrm{PE}}\to\infty .
\end{split}
\label{eq:BR}
\end{equation}
Hence the Chernoff bound~\eqref{eq:chernoff} is conservative by the factor
$\tfrac{r^\star-1}{r^\star}\sqrt{\pi N_{\mathrm{PE}}}$, growing as
$\sqrt{N_{\mathrm{PE}}}$, and the effective exponent obeys
$-N_{\mathrm{PE}}^{-1}\ln\varepsilon_{\mathrm{PE}}
=\psi(r^\star)+\tfrac{\ln N_{\mathrm{PE}}}{2N_{\mathrm{PE}}}
+O(N_{\mathrm{PE}}^{-1})$, approaching the Stein exponent from above.
\end{corollary}

\begin{proof}
The event is $\{\tfrac1N\chi^2_N\le t\}$ with $t=1/r^\star<1$
($N\equiv N_{\mathrm{PE}}$). With cumulant generating function
$\Lambda(\theta)=-\tfrac12\ln(1-2\theta)$, the exponential tilt solving
$\Lambda'(\theta^\star)=t$ is $\theta^\star=\tfrac12(1-1/t)<0$, with
$\Lambda''(\theta^\star)=2t^2$. The Bahadur--Rao expansion
$\Pr[\tfrac1N\chi^2_N\le t]=
\big(|\theta^\star|\sqrt{2\pi N\,\Lambda''(\theta^\star)}\big)^{-1}
e^{-N I(t)}(1+o(1))$, with rate $I(t)=\tfrac12(t-1-\ln t)=\psi(r^\star)$ and
$|\theta^\star|\sqrt{\Lambda''(\theta^\star)}=\tfrac{1-t}{2t}\cdot t\sqrt2
=\tfrac{1-t}{\sqrt2}$, gives the prefactor
$\big((1-t)\sqrt{\pi N}\big)^{-1}=r^\star/\big((r^\star-1)\sqrt{\pi N}\big)$.
The effective-exponent expansion follows by taking $-\tfrac1N\ln$
of~\eqref{eq:BR}.
\end{proof}

\begin{remark}[Numerical validation]\label{rem:valid}
The missed-detection probability is the exact chi-squared tail
$\varepsilon_{\mathrm{PE}}=F_{\chi^2_{N_{\mathrm{PE}}}}(N_{\mathrm{PE}}/r^\star)$, so
Corollaries~\ref{cor:finite} and~\ref{cor:BR} may be verified in closed form
without sampling. Doing so confirms that the bound~\eqref{eq:chernoff} holds
at every $N_{\mathrm{PE}}$ and is conservative by the $\sqrt{N_{\mathrm{PE}}}$
factor of Corollary~\ref{cor:BR}---about $10\times$ at $N_{\mathrm{PE}}=10^3$
and $21\times$ at $N_{\mathrm{PE}}=5\times10^3$ for $r^\star=1.2$---and that
the effective exponent $-N_{\mathrm{PE}}^{-1}\ln\varepsilon_{\mathrm{PE}}$
decreases to $\psi(r^\star)$ from above; the conservatism is in the direction
of security. Monte-Carlo estimates agree until the sampling floor
$\varepsilon_{\mathrm{PE}}\lesssim 1/N_{\mathrm{trials}}$, below which the
closed-form tail should be used.
\end{remark}

\section{Composable Finite-Size Integration}
\label{sec:composable}
We now show where the forgery-detectability rate enters a complete
composable finite-size security statement. The architecture is the standard
one for CV-QKD under collective
attacks~\cite{Renner2005,Furrer2012,Leverrier2015}: a protocol is
$\varepsilon$-secure, with
$\varepsilon=\varepsilon_{\mathrm{cor}}+\varepsilon_{\mathrm{sec}}$, if it is
$\varepsilon_{\mathrm{cor}}$-correct and
$\varepsilon_{\mathrm{sec}}$-secret, and the secrecy parameter decomposes as
$\varepsilon_{\mathrm{sec}}=\varepsilon_{\mathrm{PE}}
+2\varepsilon_{\mathrm{sm}}+\bar\varepsilon$, where
$\varepsilon_{\mathrm{PE}}$ is the probability that parameter estimation
accepts a channel outside the assumed confidence region,
$\varepsilon_{\mathrm{sm}}$ is the smoothing parameter of the smooth
min-entropy, and $\bar\varepsilon$ is the leftover-hash failure
probability~\cite{Renner2005,Leverrier2015}. The two contributions of the
present framework are (a) $\varepsilon_{\mathrm{PE}}$ is bounded by the
Stein exponent of Corollary~\ref{cor:finite} at every block length, and (b)
the worst-case Holevo term against which privacy amplification must be run
is exactly the certificate reusability rate of Corollary~\ref{cor:reuse}.

Throughout this section we work in the trusted-correlation model of
Assumption~\ref{as:std}, with $m:=N_{\mathrm{PE}}$ quadrature measurements
sacrificed for parameter estimation and $n$ retained for key generation.
Define, for $t>0$,
\begin{equation}
\varphi(t):=\tfrac12\big(t-1-\ln t\big),
\qquad \text{so that } \psi(r)=\varphi(1/r).
\label{eq:phi}
\end{equation}

\subsection{A Two-Sided Acceptance Test}

Corollary~\ref{cor:finite} used the threshold $\hat N\le N_0$, whose
Type-I (false-abort) probability is $O(1)$. An operational protocol
requires a strictly positive margin. Fix $t_{\mathrm{thr}}\ge1$ and let the
verifier accept iff
\begin{equation}
\hat N\;\le\;t_{\mathrm{thr}}\,N_0 .
\label{eq:test}
\end{equation}

\begin{lemma}[Two-sided exponents]\label{lem:twosided}
For the test~\eqref{eq:test} in the trusted-correlation model:
\begin{enumerate}
\item[(a)] \emph{(Robustness.)} Under honest operation ($N^\star=N_0$), if
$t_{\mathrm{thr}}>1$ the abort probability obeys
$\varepsilon_{\mathrm{rob}}\le\exp\!\big(-m\,\varphi(t_{\mathrm{thr}})\big)$.
\item[(b)] \emph{(Missed detection.)} Under an attack with
$r^\star=N^\star/N_0>t_{\mathrm{thr}}$, the acceptance probability obeys
\begin{equation}
\varepsilon_{\mathrm{PE}}(r^\star)\;\le\;
\exp\!\Big(-m\,\psi\big(r^\star/t_{\mathrm{thr}}\big)\Big).
\label{eq:twosided}
\end{equation}
\end{enumerate}
\end{lemma}

\begin{proof}
Both statements are the two tails of the same chi-squared Chernoff bound.
Since $m\hat N/N^\star\sim\chi^2_m$, acceptance is the event
$\tfrac1m\chi^2_m\le t_{\mathrm{thr}}N_0/N^\star=t_{\mathrm{thr}}/r^\star$
and abort under honest operation is
$\tfrac1m\chi^2_m> t_{\mathrm{thr}}$. The bounds
$\Pr[\tfrac1m\chi^2_m\le t]\le e^{-m\varphi(t)}$ for $t<1$ and
$\Pr[\tfrac1m\chi^2_m\ge t]\le e^{-m\varphi(t)}$ for $t>1$ give (a) with
$t=t_{\mathrm{thr}}$ and (b) with $t=t_{\mathrm{thr}}/r^\star$, using
$\varphi(t_{\mathrm{thr}}/r^\star)=\psi(r^\star/t_{\mathrm{thr}})$
from~\eqref{eq:phi}.
\end{proof}

\subsection{Worst-Case Leakage Compatible with Acceptance}

\begin{definition}[Worst-case parameter and leakage]\label{def:wc}
For a target $\varepsilon_{\mathrm{PE}}\in(0,1)$, let
$\psi^{-1}$ denote the inverse of $\psi$ on $[1,\infty)$ and set
\begin{align}
r_{\mathrm{wc}}(m,\varepsilon_{\mathrm{PE}})
&:= t_{\mathrm{thr}}\,
\psi^{-1}\!\Big(\tfrac1m\ln\tfrac1{\varepsilon_{\mathrm{PE}}}\Big),
\label{eq:rwc}\\
\chi_{\mathrm{wc}}(m,\varepsilon_{\mathrm{PE}})
&:= \chi_{BE}\big(T_0,\xi_{\mathrm{wc}}\big),
\quad
N(T_0,\xi_{\mathrm{wc}})=r_{\mathrm{wc}}N_0 ,
\label{eq:chiwc}
\end{align}
whenever $r_{\mathrm{wc}}$ lies in the physical range
[$\xi_{\mathrm{wc}}\le\xi_{\max}$]; otherwise set
$\chi_{\mathrm{wc}}:=\chi_{\max}^{\mathcal M}$.
\end{definition}

By Lemma~\ref{lem:twosided}(b), every attack with
$r^\star\ge r_{\mathrm{wc}}$ is detected except with probability at most
$\varepsilon_{\mathrm{PE}}$; by the strict monotonicity of
$\chi_{BE}(T_0,\cdot)$ (Proposition~\ref{prop:mono}), every attack that
evades this guarantee satisfies
$\chi_{BE}\le\chi_{\mathrm{wc}}$. For $t_{\mathrm{thr}}=1$ the quantity
$\chi_{\mathrm{wc}}$ is precisely the certificate reusability rate
$\chi_{BE}^{\max}(m,\varepsilon_{\mathrm{PE}})$ of
Corollary~\ref{cor:reuse}: the reusability rate \emph{is} the worst-case
Holevo term of the composable analysis.

\subsection{Composable Key-Length Statement}

\begin{theorem}[Composable integration, trusted-correlation
model]\label{thm:composable}
Consider the GG02 protocol of Sec.~\ref{sec:setup} with acceptance
test~\eqref{eq:test}, under collective Gaussian attacks and
Assumption~\ref{as:std}, with parameters
$\varepsilon_{\mathrm{PE}},\varepsilon_{\mathrm{sm}},\bar\varepsilon,
\varepsilon_{\mathrm{cor}}\in(0,1)$. If privacy amplification extracts
\begin{align}
\ell \;\le\;& \;
n\big[\beta I_{AB}(T_0,\xi_0)-\chi_{\mathrm{wc}}(m,\varepsilon_{\mathrm{PE}})\big]
\nonumber\\
&-\sqrt{n}\,\Delta_{\mathrm{AEP}}(\varepsilon_{\mathrm{sm}},d)
-2\log_2\tfrac1{2\bar\varepsilon}-\log_2\tfrac1{\varepsilon_{\mathrm{cor}}}
\label{eq:keylength}
\end{align}
bits, where $\Delta_{\mathrm{AEP}}(\varepsilon_{\mathrm{sm}},d)$ is the
asymptotic-equipartition correction for the $d$-dimensional discretized
key map of Ref.~\cite{Leverrier2015}, then the protocol is
$\varepsilon_{\mathrm{cor}}$-correct and
$\varepsilon_{\mathrm{sec}}$-secret with
\begin{equation}
\varepsilon_{\mathrm{sec}}
=\varepsilon_{\mathrm{PE}}+2\varepsilon_{\mathrm{sm}}+\bar\varepsilon ,
\label{eq:epssec}
\end{equation}
hence $\varepsilon$-secure with
$\varepsilon=\varepsilon_{\mathrm{cor}}+\varepsilon_{\mathrm{sec}}$.
Its robustness obeys
$\varepsilon_{\mathrm{rob}}\le\exp(-m\,\varphi(t_{\mathrm{thr}}))$.
\end{theorem}

\begin{proof}
Correctness with parameter $\varepsilon_{\mathrm{cor}}$ is guaranteed by the
hash-comparison step of the error-correction verification, exactly as in
Refs.~\cite{Furrer2012,Leverrier2015}; it is independent of the PE analysis.
For secrecy, condition on acceptance of the test~\eqref{eq:test} and split
the attack space along $r^\star\lessgtr r_{\mathrm{wc}}$.

\emph{Case 1: $r^\star\ge r_{\mathrm{wc}}$.} By
Lemma~\ref{lem:twosided}(b) and the definition~\eqref{eq:rwc}, acceptance
occurs with probability at most
$\exp(-m\,\psi(r_{\mathrm{wc}}/t_{\mathrm{thr}}))
=\varepsilon_{\mathrm{PE}}$.
This event contributes at most $\varepsilon_{\mathrm{PE}}$ to the trace
distance defining secrecy, exactly as the PE-failure event does in the
standard decomposition~\cite{Leverrier2015}.

\emph{Case 2: $r^\star<r_{\mathrm{wc}}$.} The true channel lies in the
confidence region. Under collective Gaussian attacks the state shared after
$n$ rounds is an i.i.d.\ product, and by Gaussian
extremality~\cite{Wolf2006,GarciaPatron2006} Eve's information per symbol is
bounded by the Holevo quantity evaluated at the true covariance; by the
strict monotonicity of $\chi_{BE}(T_0,\cdot)$
(Proposition~\ref{prop:mono}) this is at most
$\chi_{\mathrm{wc}}$. The finite-size asymptotic equipartition property for
the smooth min-entropy~\cite{Renner2005,Leverrier2015} then gives
$H_{\min}^{\varepsilon_{\mathrm{sm}}}(\bm y\mid E)\ge
n\,[H(y)-\chi_{\mathrm{wc}}]-\sqrt{n}\,
\Delta_{\mathrm{AEP}}(\varepsilon_{\mathrm{sm}},d)$ for the discretized key
variable, and subtracting the error-correction leakage
[$\le n\,H(y)-n\,\beta I_{AB}$ by definition of $\beta$, verified
operationally; see Remark~\ref{rem:leakEC}] and applying the leftover-hash
lemma~\cite{Renner2005} with failure parameter $\bar\varepsilon$ yields
$\varepsilon_{\mathrm{sec}}$-secrecy of any key of
length~\eqref{eq:keylength} with
$2\varepsilon_{\mathrm{sm}}+\bar\varepsilon$ from this case.

Summing the two contributions gives~\eqref{eq:epssec}. Robustness is
Lemma~\ref{lem:twosided}(a).
\end{proof}

\begin{corollary}[Recovery of the Gaussian confidence-interval analysis]
\label{cor:recovery}
As $\tfrac1m\ln\tfrac1{\varepsilon_{\mathrm{PE}}}\to0$,
\begin{equation}
r_{\mathrm{wc}}
= t_{\mathrm{thr}}\left(1+2\sqrt{\tfrac1m\ln\tfrac1{\varepsilon_{\mathrm{PE}}}}
+O\!\big(\tfrac1m\ln\tfrac1{\varepsilon_{\mathrm{PE}}}\big)\right),
\label{eq:rwcexp}
\end{equation}
so the worst-case conditional variance exceeds the threshold by the relative
margin $2\sqrt{\ln(1/\varepsilon_{\mathrm{PE}})/m}$. This coincides with the
inflation $z_{\varepsilon_{\mathrm{PE}}}\,\sigma_{\hat N}/N_0$ used in
standard finite-size parameter estimation~\cite{Leverrier2010}, with
$\sigma_{\hat N}^2=2N_0^2/m$ the chi-squared variance and
$z_{\varepsilon}=\sqrt{2\ln(1/\varepsilon)}$ the Gaussian-tail quantile.
Away from this limit the two treatments differ: the Gaussian-quantile
confidence interval is a tail \emph{approximation}, whereas
$r_{\mathrm{wc}}$ derives from the non-asymptotic
bound~\eqref{eq:twosided}, which is valid at every $m$. Since
$\psi(r)<\tfrac14(r-1)^2$ for $r>1$, one has
$\psi^{-1}(u)>1+2\sqrt{u}$, so the rigorous worst-case parameter is
strictly larger than its Gaussian extrapolation; the discrepancy is the
price of replacing an asymptotic confidence statement with a bound that
holds at all block lengths, and it vanishes in the regime
$m\gg\ln(1/\varepsilon_{\mathrm{PE}})$ of practical interest.
\end{corollary}

\begin{proof}
From Lemma~\ref{lem:uni}, $\psi(r)=\tfrac14(r-1)^2-\tfrac13(r-1)^3+O((r-1)^4)$
near $r=1$, so $\psi^{-1}(u)=1+2\sqrt u+O(u)$ as $u\downarrow0$;
substituting $u=\tfrac1m\ln\tfrac1{\varepsilon_{\mathrm{PE}}}$
into~\eqref{eq:rwc} gives~\eqref{eq:rwcexp}. The identification with the
Gaussian confidence interval follows from
$z_{\varepsilon}\sigma_{\hat N}/N_0
=\sqrt{2\ln(1/\varepsilon)}\cdot\sqrt{2/m}
=2\sqrt{\ln(1/\varepsilon)/m}$. The strict inequality
$\psi(r)<\tfrac14(r-1)^2$ for $r>1$ follows since
$\tfrac14(r-1)^2-\psi(r)$ vanishes to second order at $r=1$ with strictly
positive third derivative $\ge \tfrac{1}{r^3}-\ldots>0$ there; explicitly,
$\tfrac{d}{dr}[\tfrac14(r-1)^2-\psi(r)]
=\tfrac{r-1}2-\tfrac{r-1}{2r^2}=\tfrac{(r-1)(r^2-1)}{2r^2}>0$ for $r>1$.
\end{proof}

\begin{remark}[Operational error-correction leakage]\label{rem:leakEC}
In implementations the term $n\,\beta I_{AB}$ in~\eqref{eq:keylength} is
not assumed but measured: the syndrome length
$\mathrm{leak}_{\mathrm{EC}}$ actually disclosed during reconciliation
replaces $n[H(y)-\beta I_{AB}]$, and correctness is certified by the hash
comparison at level $\varepsilon_{\mathrm{cor}}$. This is the standard
operational reading of~\eqref{eq:keylength}~\cite{Leverrier2015,Jouguet2013}
and does not interact with the PE analysis.
\end{remark}

\begin{remark}[What remains outside the trusted-correlation
statement]\label{rem:outside}
Three ingredients of a fully general proof are additional to
Theorem~\ref{thm:composable} and orthogonal to the exponent it quantifies.
(1)~\emph{Coherent attacks:} the reduction to collective Gaussian attacks
proceeds via the postselection technique~\cite{Christandl2009} or the
Gaussian de~Finetti reduction with its energy
test~\cite{Renner2007,Leverrier2017}, degrading $\varepsilon$ by a
polynomial dilution factor; this affects the prefactor of the security
parameter but not the exponent~\eqref{eq:twosided}.
(2)~\emph{Monitoring noise:} the trusted-correlation model idealizes the
monitored quantities $(V_A,c)$ as known exactly; finite monitoring precision
adds a union-bound term to $\varepsilon_{\mathrm{PE}}$ and shrinks
$t_{\mathrm{thr}}$ by the monitoring tolerance, leaving the structure
of~\eqref{eq:keylength} unchanged. (3)~\emph{General monitored sets:} for
$\mathcal{M}$ beyond the trusted-correlation model, the scalar statistic
$\hat N$ is replaced by a Wishart-type statistic
(cf.\ Remark~\ref{rem:scope}) and $\chi_{\mathrm{wc}}$ by
$g_{\mathcal M}^{-1}$ via Theorem~\ref{thm:main}(iv) under
Condition~\ref{cond:reg}; the exponent remains
$D(\Sigma_0\Vert\Sigma_1)$ with modified one-shot constants.
\end{remark}

\section{Convexity}
\label{sec:convexity}

\begin{theorem}[Convexity on the operational range]\label{thm:convex}
In the trusted-correlation model, suppose that $\chi_{BE}(T_0,\cdot)$ is
concave and strictly increasing in $\xi$ and that $r^\star<2$. Then
$g_{\mathcal M}$ is convex in $\chi_0$ on
$[\chi_{\mathrm{ben}},\chi_{\mathrm{null}}]$.
\end{theorem}

\begin{proof}
Let $\chi(\xi):=\chi_{BE}(T_0,\xi)$. By assumption, $\chi$ is concave and
strictly increasing on the relevant domain. Let $\xi(\chi)$ denote its
inverse; this inverse is convex and strictly increasing.

Define $r(\chi):=N(\xi(\chi))/N_0$. Since $N$ depends linearly on $\xi$,
the map $\chi\mapsto r(\chi)$ is convex and strictly increasing.

From Lemma~\ref{lem:uni}, the function $\psi(r)$ is convex and strictly
increasing on the interval $[1,2)$ [because $\psi''(r)>0$ for $r<2$].

Therefore $g_{\mathcal M}(\chi_0)=\psi(r(\chi_0))$ is the composition of the
convex increasing function $\psi$ with the convex function $r(\chi_0)$.
Such a composition is convex. Hence $g_{\mathcal M}$ is convex in $\chi_0$
on the interval where $r^\star<2$.

The assumption $r^\star<2$ holds throughout the operational regime of
positive key rate (the inflection point of $\psi$ at $r=2$ lies beyond any
regime with $K>0$).
\end{proof}

\begin{proposition}[Concavity of the leakage functionals]\label{prop:convex}
Fixing $T = T_0$, the covariance $\Sigma_{AB}(\xi)$ is affine in $\xi$, and
over the parameter range where key positivity holds:
\begin{enumerate}
\item[(a)] Eve's entropy $S_E(\xi) = S(\rho_{AB})$ is strictly increasing
and concave in $\xi$;
\item[(b)] $\chi_{BE}(T_0, \cdot)$ is concave in $\xi$.
\end{enumerate}
Statement (a) is proved: strict monotonicity follows from strict concavity
of the von Neumann entropy under the displacement mixture
$\rho\mapsto\int d\mu_s(z)D(z)\rho D(z)^\dagger$ of
Proposition~\ref{prop:mono} (the displaced states are distinct), and
concavity in $\xi$ follows from the quantum de~Bruijn identity along the
additive-noise semigroup $\{\mathcal M_s\}_{s\ge0}$ on mode $B$,
$\partial_s S(\rho_s)\propto J(\rho_s)$ with $J$ the displacement Fisher
information of mode $B$, together with the monotonicity of $J$ under the
displacement-covariant channel $\mathcal M_u$ [data
processing]~\cite{KonigSmith2014}; since $b$ is affine in $\xi$, $S_E$ is
concave in $\xi$. Statement (b) does not admit a na\"ive
``concave $-$ convex'' decomposition---by (a) and the concavity of
$b\mapsto g(\nu_c(b))$, the quantity $\chi_{BE}=S_E-g(\nu_c)$ is a
difference of two concave functions---and is established numerically.
[Strict monotonicity of $\chi_{BE}(T_0,\cdot)$, formerly grouped with this
statement, is proved analytically in Proposition~\ref{prop:mono}.]
\end{proposition}

\begin{remark}[Epistemic status]\label{rem:status}
Theorem~\ref{thm:main}(i)--(ii) hold for general $\mathcal{M}$.
Proposition~\ref{prop:mono}, Theorem~\ref{thm:main}(iii),
Corollary~\ref{cor:finite}, Corollary~\ref{cor:reuse},
Lemma~\ref{lem:twosided}, Theorem~\ref{thm:composable}, and
Proposition~\ref{prop:convex}(a) are proved unconditionally in the
trusted-correlation model; in particular, the composable statement rests on
no numerical input. Theorem~\ref{thm:main}(iv) is proved conditionally on
Condition~\ref{cond:reg}, which is itself proved in the trusted-correlation
model (Proposition~\ref{prop:condTC}) and open for general $\mathcal{M}$.

The sole remaining numerical ingredient is the concavity of
$\chi_{BE}(T_0,\cdot)$ in Proposition~\ref{prop:convex}(b), which enters
only Theorem~\ref{thm:convex}; an analytic proof would make convexity of
$g_{\mathcal M}$ unconditional in the operational regime.
\end{remark}

\begin{remark}[Scope]\label{rem:scope}
Corollary~\ref{cor:finite} is the collective-attack exponent for a homodyne
single-quadrature estimator; heterodyne or jointly estimated correlation
replaces the $\chi^2_{N_{\mathrm{PE}}}$ statistic by a Wishart one (with a
degrees-of-freedom reduction equal to the number of jointly estimated
covariance parameters), retaining exponent $D(\Sigma_0\Vert\Sigma_1)$ with
larger one-shot constants. The Gaussian de~Finetti
reduction~\cite{Renner2007,Leverrier2017}, entropy smoothing, and
reconciliation/privacy-amplification terms of the composable bound are
accounted for in Sec.~\ref{sec:composable}; the residual ingredients are
enumerated in Remark~\ref{rem:outside}.
\end{remark}

\section{Numerical Validation}
\label{sec:results}

\subsection{Monte Carlo Validation of Lemma~\ref{lem:uni} and
Corollary~\ref{cor:finite}}

To numerically validate the closed-form Stein exponent derived in
Lemma~\ref{lem:uni} and the rigorous non-asymptotic parameter-estimation
failure bound of Corollary~\ref{cor:finite}, we performed extensive Monte
Carlo simulations of the trusted-correlation model. For each combination of
forged conditional-variance ratio $r^\star > 1$ and block size
$N_{\mathrm{PE}}$, we generated $60{,}000$ independent realizations of the
sample conditional variance $\hat{N}$ under the attack distribution and
computed the empirical missed-detection probability (the fraction of trials
in which $\hat{N} \le N_0$). All supplementary materials are available in
Ref.~\cite{repo}.

\begin{sloppypar}
Figure~\ref{fig:mc_bound} compares the simulated missed-detection
probability against the theoretical upper bound
$\exp(-N_{\mathrm{PE}} \psi(r^\star))$ of Corollary~\ref{cor:finite} on a
log-log scale. The empirical curves lie strictly below the theoretical
bounds for all tested values of $r^\star$ and $N_{\mathrm{PE}}$, confirming
that the non-asymptotic bound holds in practice. As $N_{\mathrm{PE}}$
increases, both the empirical and theoretical probabilities decay
exponentially, with faster decay for larger deviations $r^\star$.
\end{sloppypar}

Figure~\ref{fig:mc_exponent} shows the convergence of the empirical exponent
$-\ln(\varepsilon_{\mathrm{emp}})/N_{\mathrm{PE}}$ toward the theoretical
Stein exponent $\psi(r^\star)$ of Lemma~\ref{lem:uni}. For every $r^\star$,
the simulated exponent approaches the predicted horizontal asymptote as
$N_{\mathrm{PE}}$ grows, providing direct numerical confirmation that
$\psi(r)$ is the correct large-deviation rate.

Table~\ref{tab:mc_summary} summarizes representative results for
$N_{\mathrm{PE}} = 1000$ and $N_{\mathrm{PE}} = 5000$. In all cases the
empirical failure probability remains below the theoretical bound, and the
extracted exponent is consistent with $\psi(r^\star)$. Near the threshold
($r^\star = 1.02$) the quadratic approximation
$\psi(r) \approx \frac14(r-1)^2$ is already accurate, reproducing the
Gaussian confidence-interval scaling used in conventional finite-size
analyses (Corollary~\ref{cor:recovery}).

\begin{table}[t]
\centering
\caption{Summary of Monte Carlo validation for selected $r^\star$ and
$N_{\mathrm{PE}}$ ($60{,}000$ trials per entry).}
\label{tab:mc_summary}
\setlength{\tabcolsep}{3.5pt}
\small
\begin{tabular}{lcccc}
\toprule
$r^\star$ & $N_{\mathrm{PE}}$ & Emp.\ $\varepsilon_{\mathrm{PE}}$ &
$\exp(-N\psi)$ & Emp.\ exp. \\
\midrule
1.02 & 1000 & $3.35 \times 10^{-1}$ & $9.07 \times 10^{-1}$ & 0.0011 \\
1.02 & 5000 & $1.63 \times 10^{-1}$ & $6.14 \times 10^{-1}$ & 0.0004 \\
1.05 & 1000 & $1.44 \times 10^{-1}$ & $5.57 \times 10^{-1}$ & 0.0019 \\
1.05 & 5000 & $7.73 \times 10^{-3}$ & $5.35 \times 10^{-2}$ & 0.0010 \\
1.10 & 1000 & $1.86 \times 10^{-2}$ & $1.11 \times 10^{-1}$ & 0.0040 \\
1.10 & 5000 & $< 10^{-5}$ & $1.67 \times 10^{-5}$ & --- \\
1.20 & 1000 & $< 10^{-5}$ & $1.38 \times 10^{-7}$ & --- \\
1.20 & 5000 & $< 10^{-5}$ & $4.97 \times 10^{-35}$ & --- \\
\bottomrule
\end{tabular}
\end{table}

These simulations confirm that both the Stein exponent of
Lemma~\ref{lem:uni} and the finite-size bound of Corollary~\ref{cor:finite}
are tight and practically useful. The observed behavior is fully consistent
with the information-theoretic analysis presented in
Secs.~\ref{sec:lemma}--\ref{sec:composable}.

\begin{figure}[t]
\centering
\includegraphics[width=\columnwidth]{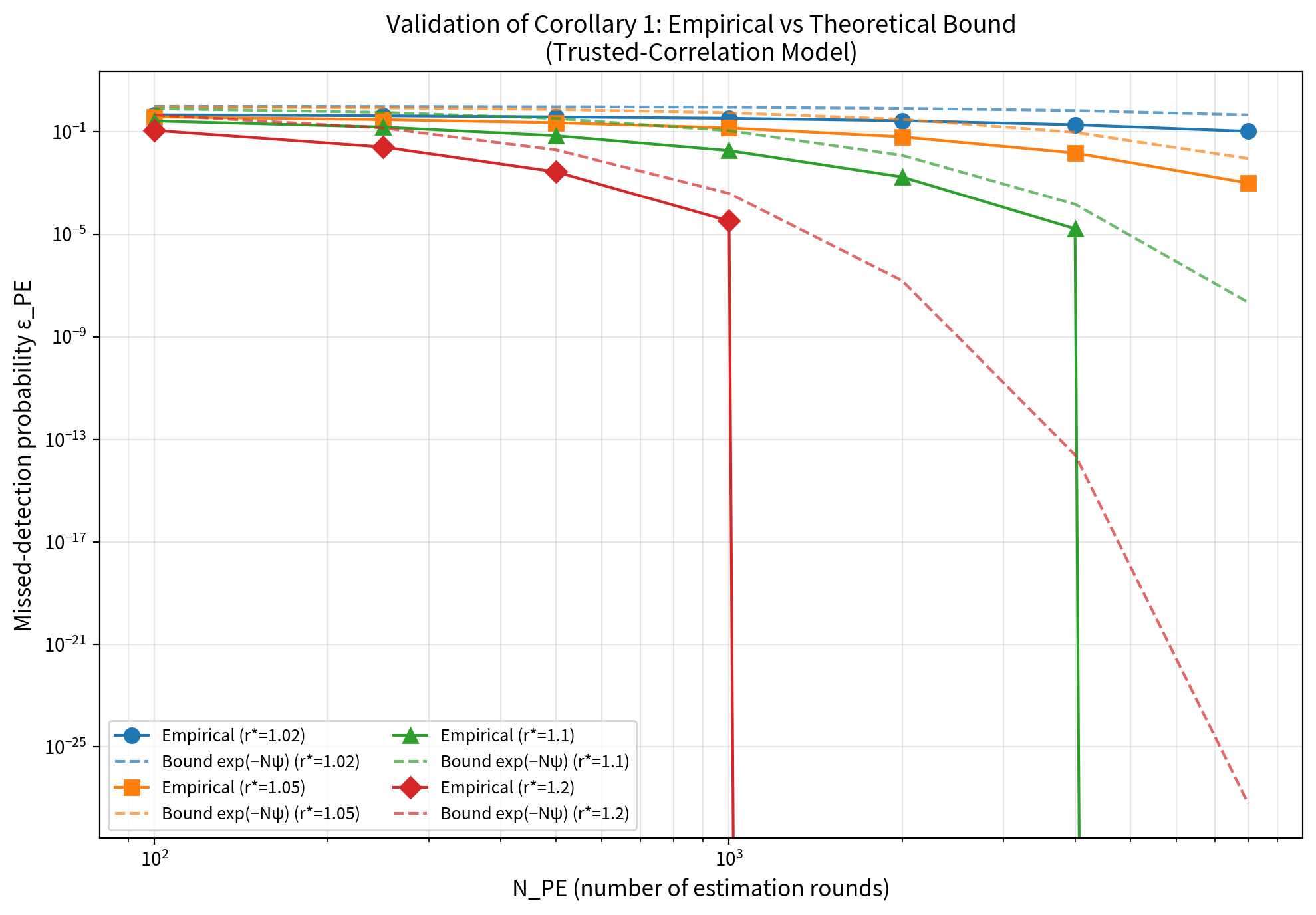}
\caption{Validation of Corollary~\ref{cor:finite}. Empirical
missed-detection probability (solid markers) versus the theoretical upper
bound $\exp(-N_{\mathrm{PE}}\psi(r^\star))$ (dashed lines) for several
forged ratios $r^\star$. The bound holds for all tested block sizes.}
\label{fig:mc_bound}
\end{figure}

\begin{figure}[t]
\centering
\includegraphics[width=\columnwidth]{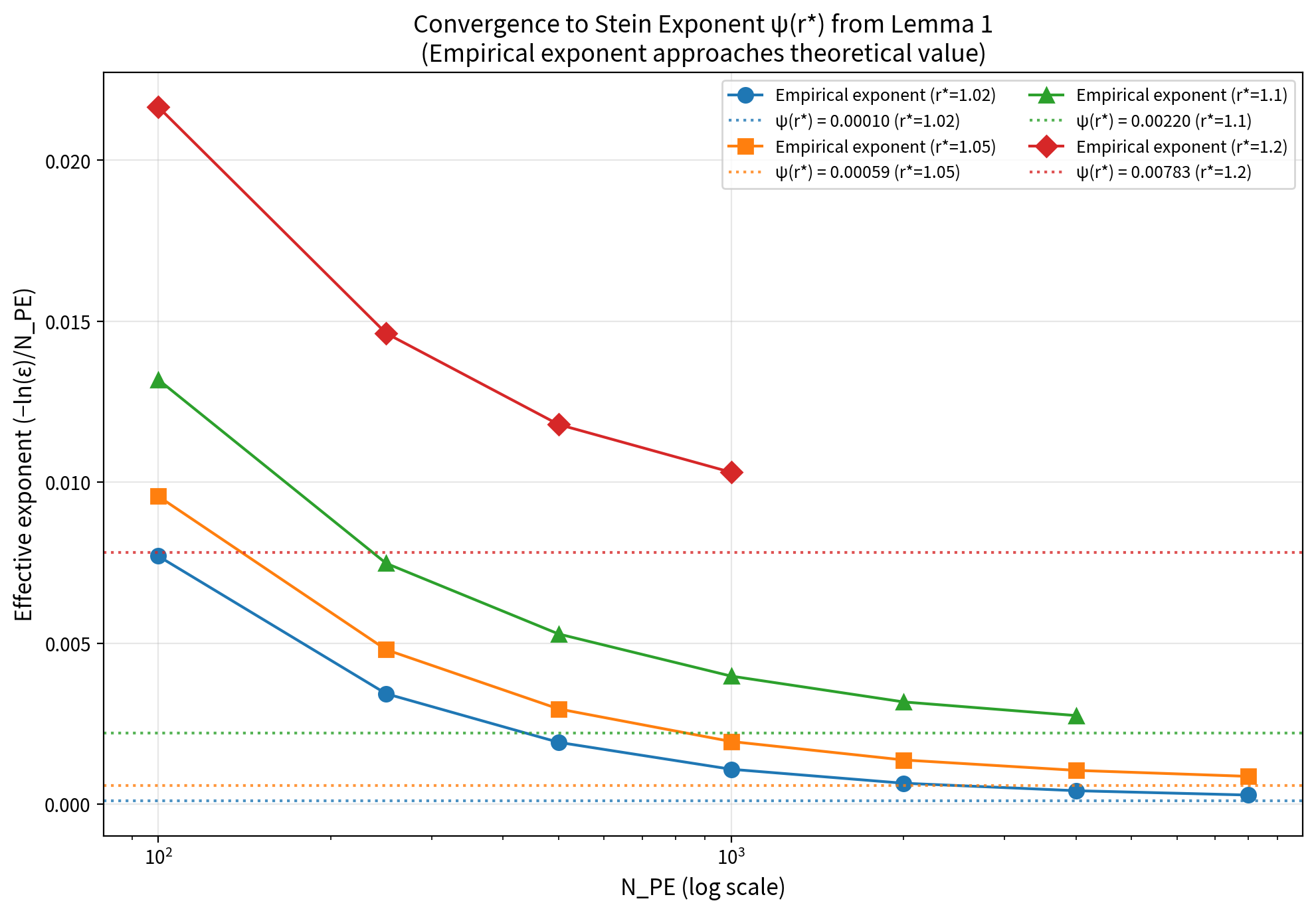}
\caption{Convergence of the empirical exponent to the Stein exponent of
Lemma~\ref{lem:uni}. For each $r^\star$, the simulated rate
$-\ln(\varepsilon)/N_{\mathrm{PE}}$ (markers) approaches the theoretical
value $\psi(r^\star)$ (dotted horizontal lines) as $N_{\mathrm{PE}}$
increases.}
\label{fig:mc_exponent}
\end{figure}

\begin{figure}[t]
\centering
\includegraphics[width=\columnwidth]{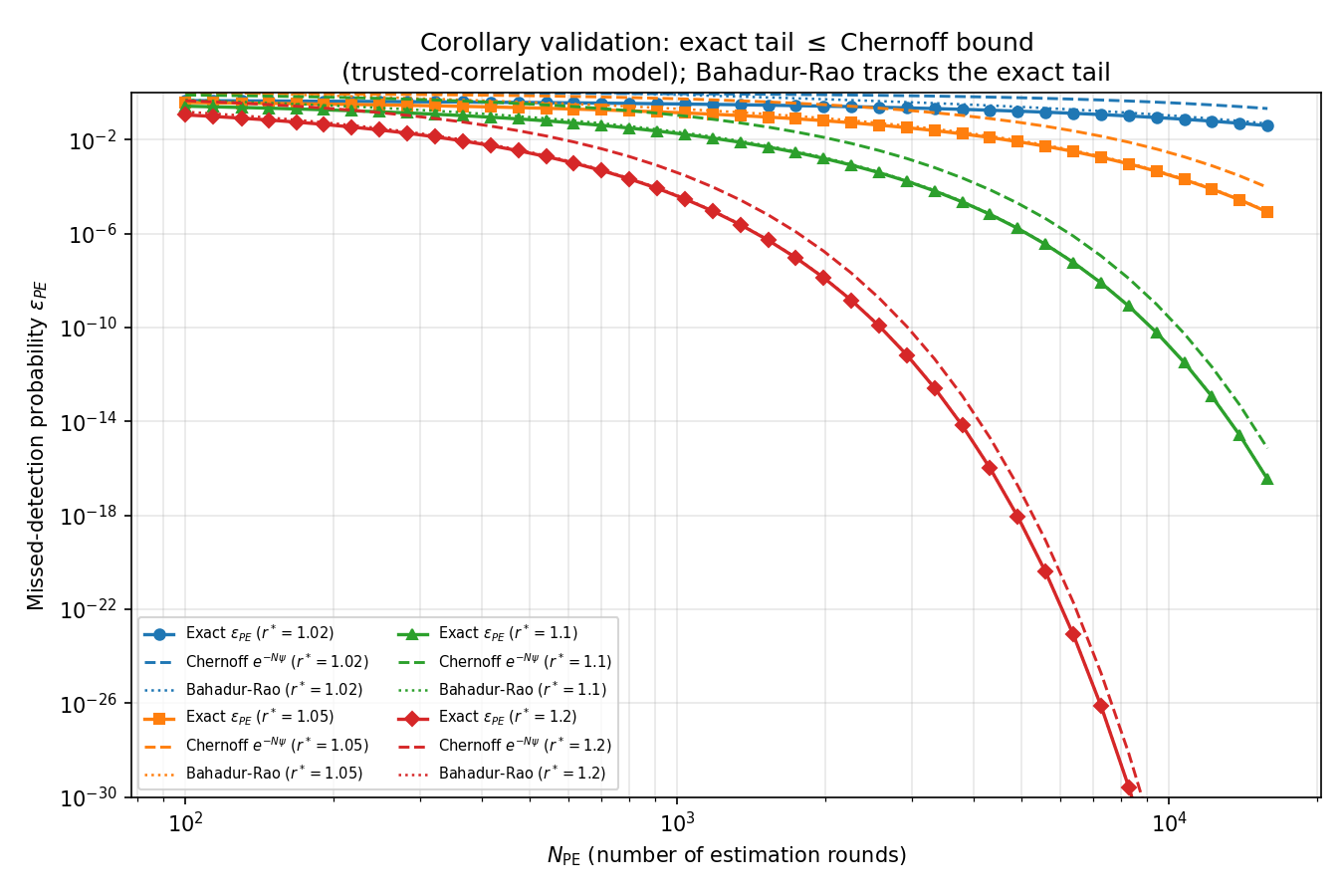}
\caption{Validation of the finite-size bound (trusted-correlation model).
Exact missed-detection probability
$\varepsilon_{\mathrm{PE}}=F_{\chi^2_{N_{\mathrm{PE}}}}(N_{\mathrm{PE}}/r^\star)$
(solid) against the Chernoff bound $e^{-N_{\mathrm{PE}}\psi(r^\star)}$ of
Corollary~\ref{cor:finite} (dashed) and the Bahadur--Rao refinement of
Corollary~\ref{cor:BR} (dotted), for attack strengths $r^\star=N^\star/N_0$.
The exact tail lies below the Chernoff bound at every $N_{\mathrm{PE}}$, the
gap growing as $\sqrt{N_{\mathrm{PE}}}$ (conservative, in the direction of
security); the Bahadur--Rao curve tracks the exact tail. All curves are
closed-form, with no Monte-Carlo sampling floor.}
\label{fig:bound}
\end{figure}

\begin{figure}[t]
\centering
\includegraphics[width=\columnwidth]{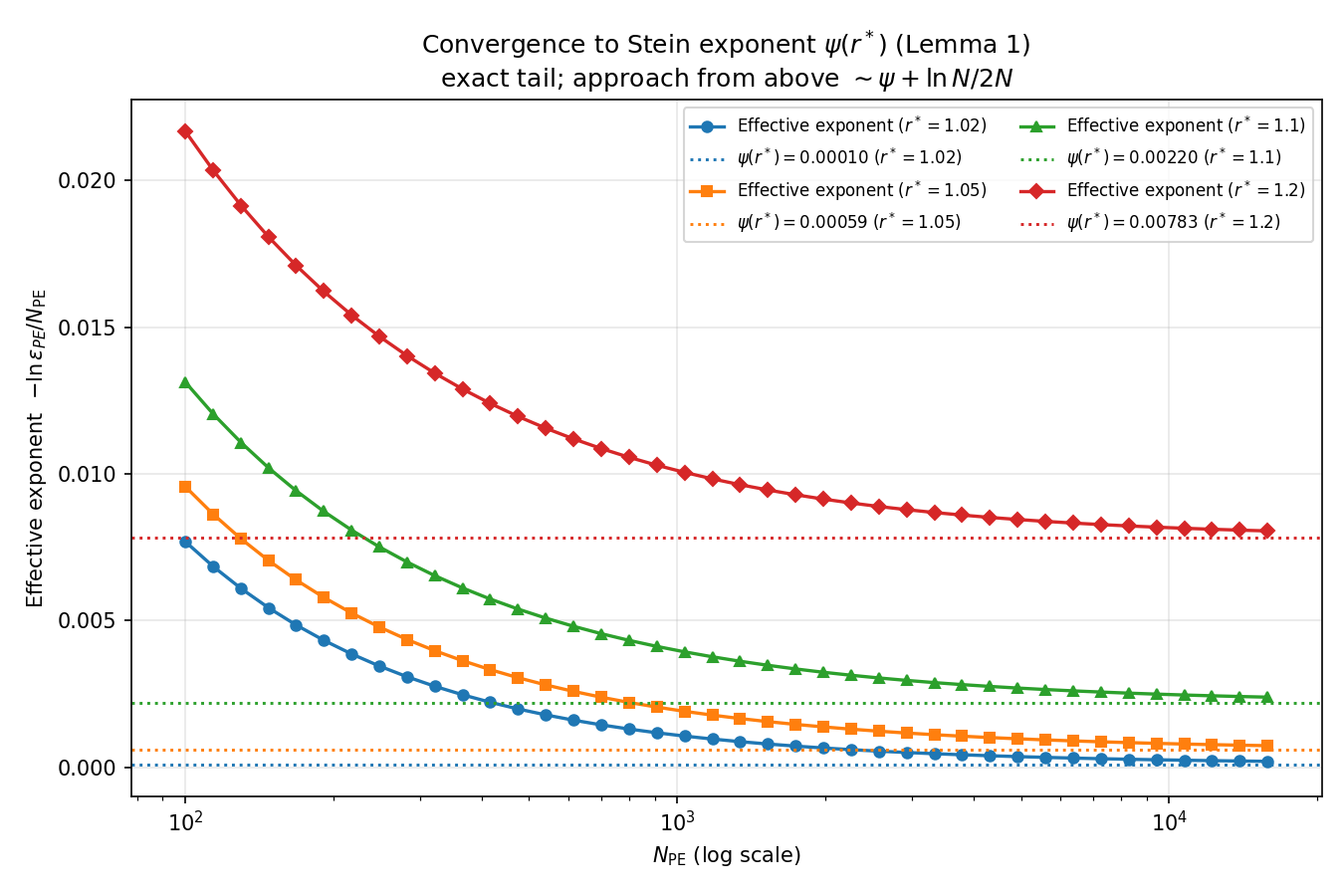}
\caption{Convergence of the effective exponent
$-N_{\mathrm{PE}}^{-1}\ln\varepsilon_{\mathrm{PE}}$ (solid) to the Stein
exponent $\psi(r^\star)$ of Lemma~\ref{lem:uni} (dotted). The effective
exponent approaches $\psi(r^\star)$ from above at rate
$\ln N_{\mathrm{PE}}/(2N_{\mathrm{PE}})$, as predicted by
Corollary~\ref{cor:BR}; convergence is slower for near-threshold $r^\star$,
where $\psi$ is small.}
\label{fig:conv}
\end{figure}

\begin{figure}[t]
\centering
\includegraphics[width=\columnwidth]{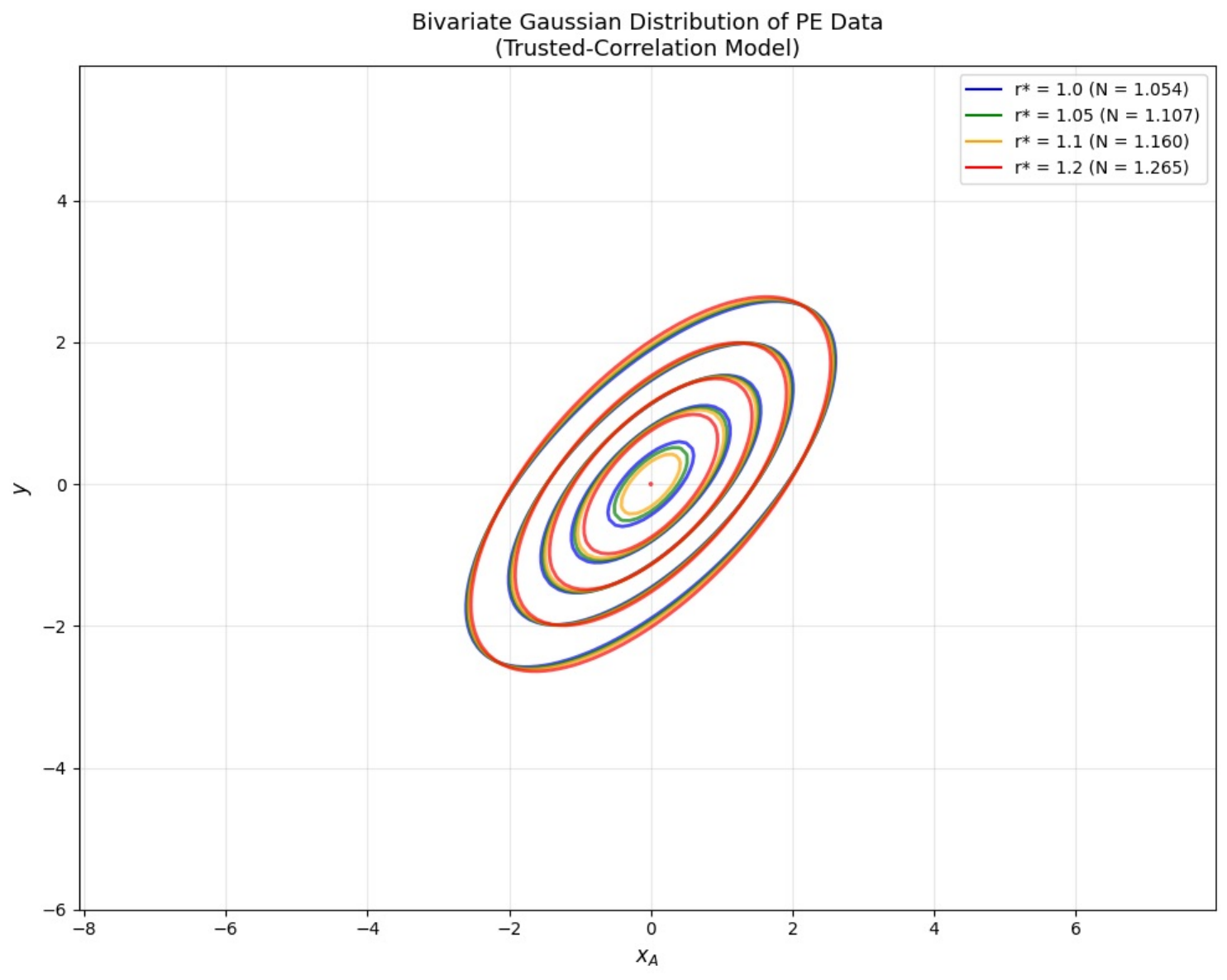}
\caption{Bivariate Gaussian distribution of the parameter-estimation data
$(x_A, y)$ in the trusted-correlation model for different values of the
forged conditional-variance ratio $r^\star = N/N_0$. The contours represent
level sets of the joint probability density. As $r^\star$ increases, the
distribution becomes wider along the $y$-direction while the correlation
(slope of the ellipse) remains fixed, making the forged statistics
progressively more distinguishable from the benign case ($r^\star = 1$).}
\label{fig:bivariate_gaussian_pe}
\end{figure}

\begin{figure*}[t]
\centering
\subfloat[Benign state ($r \approx 1$)\label{fig:wigner_benign}]{%
    \includegraphics[width=0.42\textwidth]{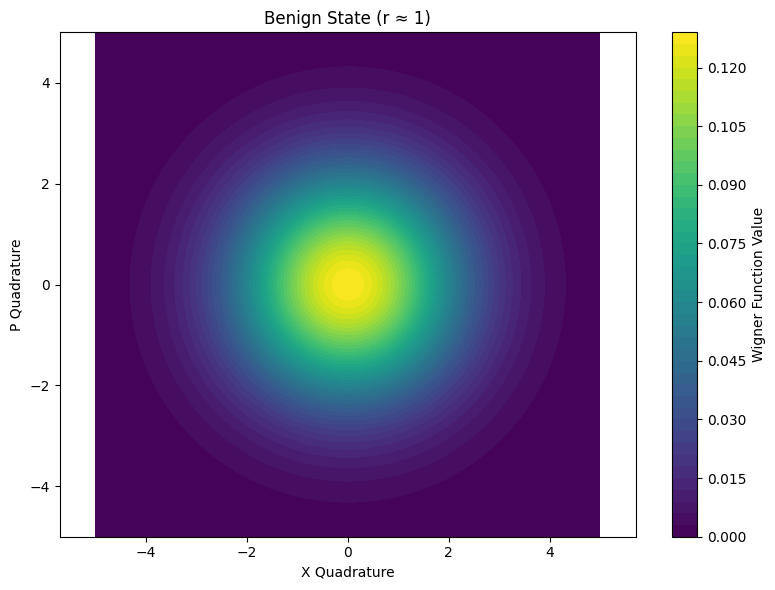}}
\hfil
\subfloat[Forged state (higher excess noise)\label{fig:wigner_forged}]{%
    \includegraphics[width=0.42\textwidth]{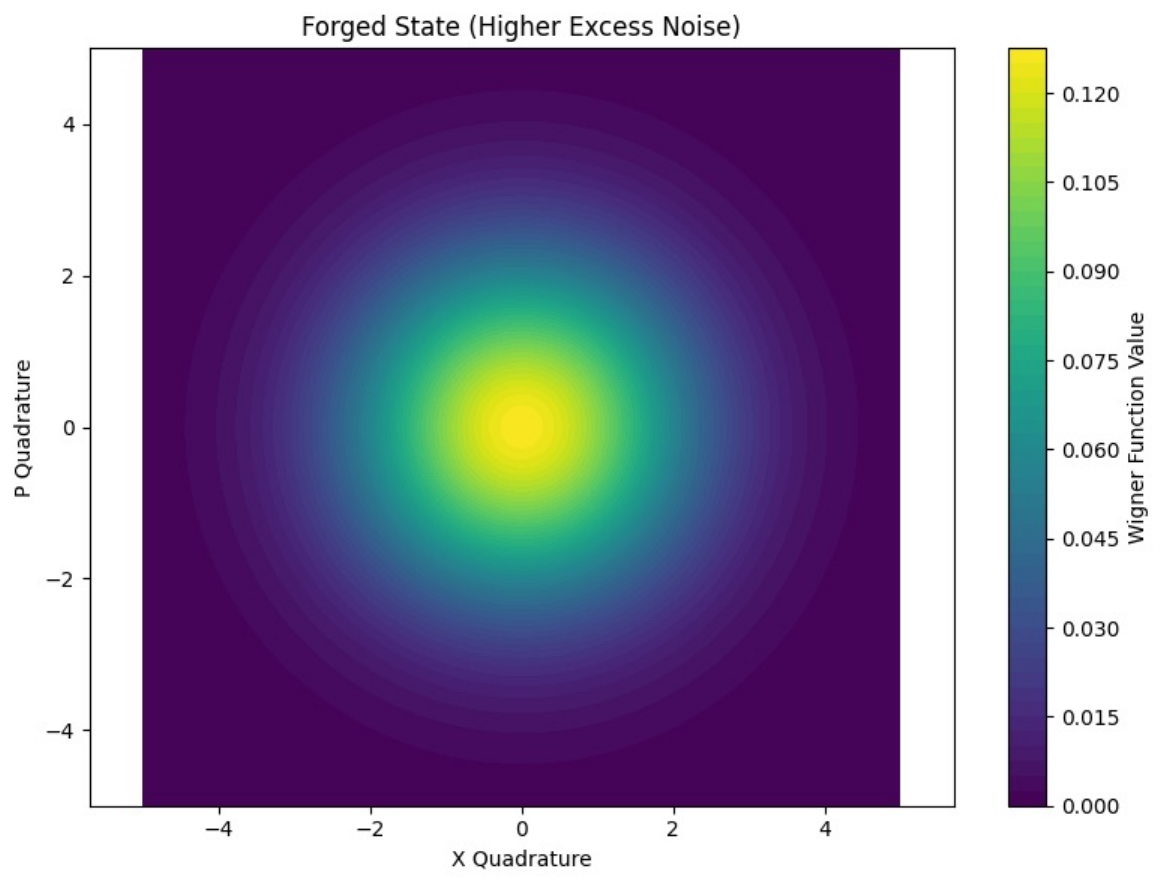}}
\caption{Wigner functions of the Gaussian quantum states corresponding to
the benign operating point and a forged state with increased excess noise.
Both states are centered at the origin in phase space. The forged state
exhibits a broader distribution due to higher excess noise, which increases
the distinguishability from the benign state and leads to a positive Stein
exponent $\psi(r^\star)$ (see Lemma~\ref{lem:uni}).}
\label{fig:wigner_functions}
\end{figure*}

\begin{figure*}[t]
\centering
\subfloat[Stein exponent $\psi(r)$ over a wide range of
$r$\label{fig:psi_wide}]{%
    \includegraphics[width=0.48\textwidth]{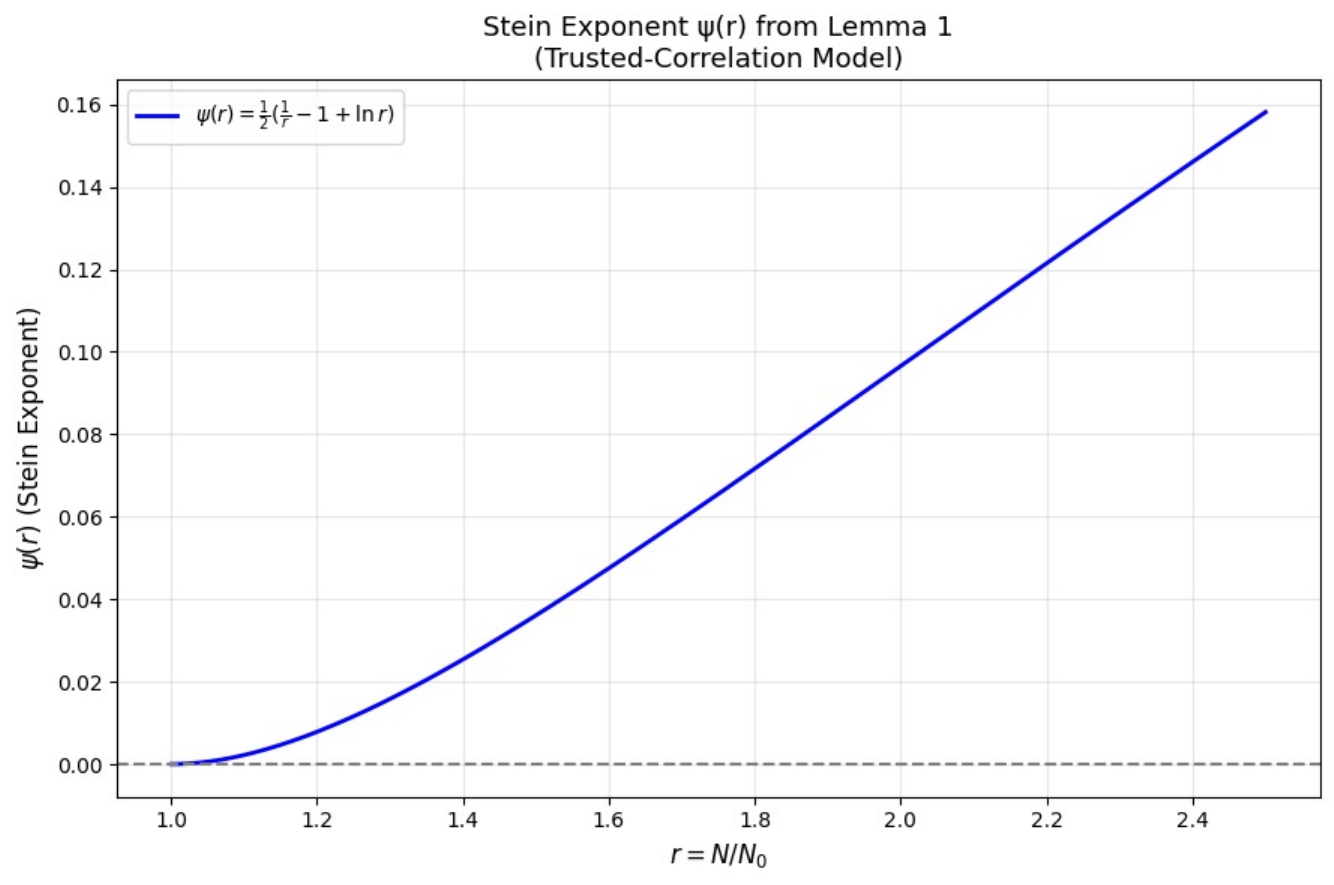}}
\hfil
\subfloat[Near-threshold behavior and quadratic
approximation\label{fig:psi_near}]{%
    \includegraphics[width=0.48\textwidth]{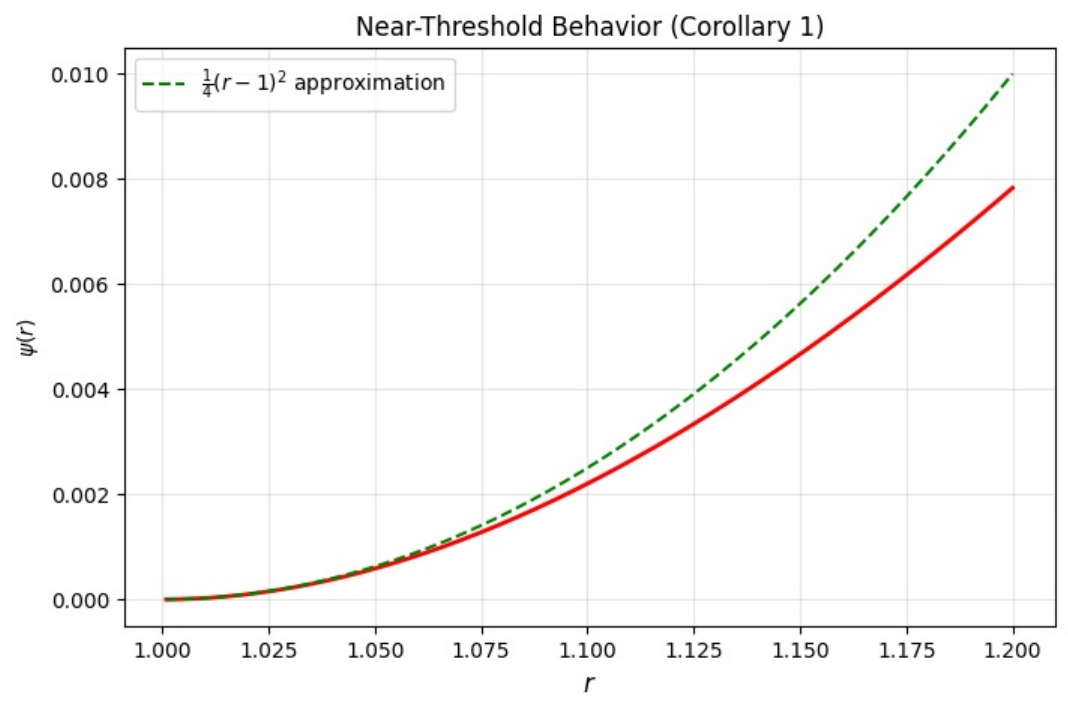}}
\caption{Stein exponent $\psi(r)$ from Lemma~\ref{lem:uni} in the
trusted-correlation model. (a) Behavior of $\psi(r)$ for $r \in [1, 2.5]$.
(b) Zoomed-in view near the threshold ($r \approx 1$), showing that
$\psi(r)$ closely follows the quadratic approximation $\frac{1}{4}(r-1)^2$,
which explains the Gaussian confidence-interval scaling used in conventional
finite-size analyses (Corollary~\ref{cor:recovery}).}
\label{fig:stein_exponent}
\end{figure*}

\begin{figure*}[t]
\centering
\includegraphics[width=0.9\textwidth]{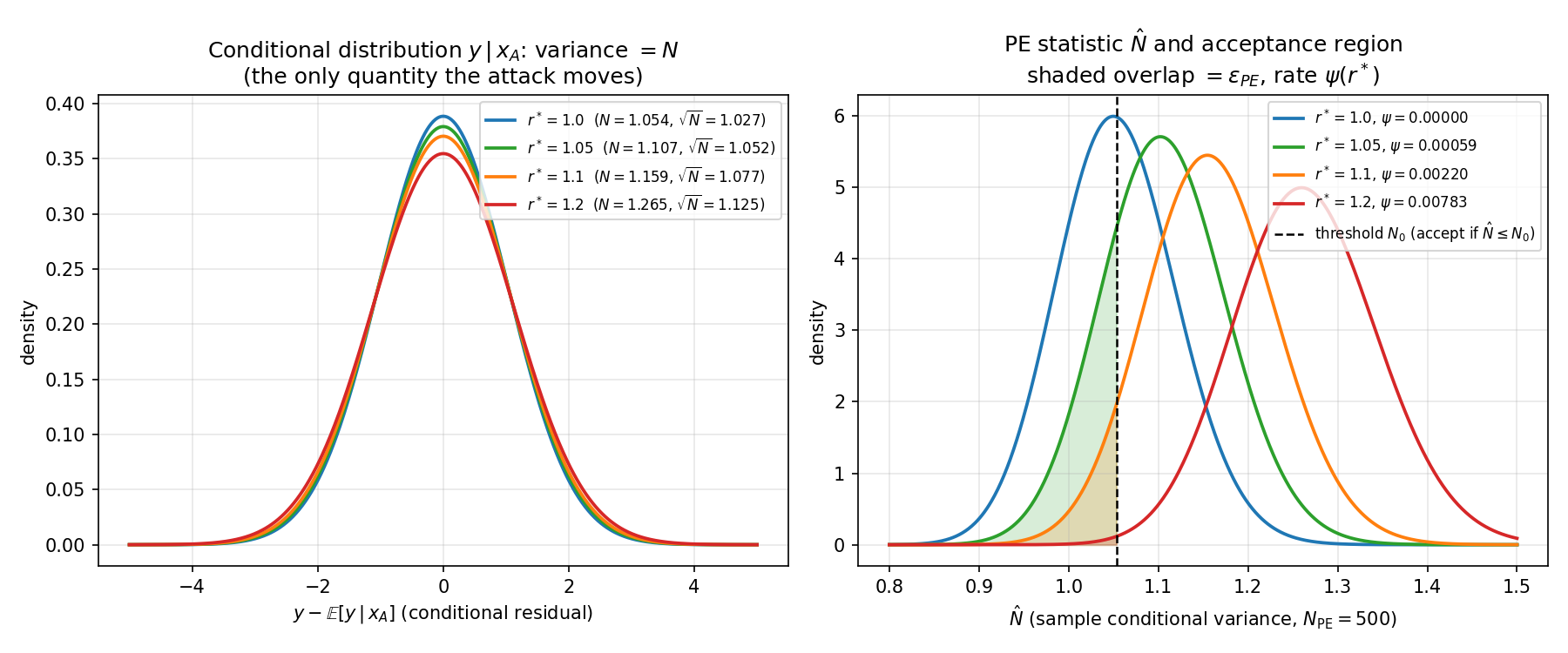}
\caption{Geometry of the trusted-correlation detection problem.
\emph{Left:} the conditional distribution $y\,|\,x_A$ is Gaussian with
variance $N$---the only quantity an attack moves. The widths $\sqrt{N}$
grow only from $1.027$ to $1.125$ as $r^\star=N^\star/N_0$ rises from $1$
to $1.2$, the geometric statement that benign and attacked channels are
nearly indistinguishable. \emph{Right:} sampling distribution of the
parameter-estimation statistic $\hat N$ at $N_{\mathrm{PE}}=500$, with the
acceptance region $\{\hat N\le N_0\}$ (dashed). For each attack the shaded
mass below the threshold is the missed-detection probability
$\varepsilon_{\mathrm{PE}}$, whose exponential decay rate is the Stein
exponent $\psi(r^\star)$ of Lemma~\ref{lem:uni} (legend); the overlap
shrinks with $r^\star$ at the rate quantified by
Corollaries~\ref{cor:finite}--\ref{cor:BR}.}
\label{fig:geometry}
\end{figure*}

\section{Discussion}
\label{sec:discussion}
The dichotomy in Theorem~\ref{thm:main}---$g_{\mathcal M}>0$ under a trusted
shot-noise unit, $g_{\mathcal M}\equiv0$ without it---places the standard
real-time monitoring countermeasures~\cite{Ma2013,JouguetCalib2013,Qi2015}
on an information-theoretic footing: they are precisely the resource that
renders the benign-statistics certificate non-forgeable, and the
rate~\eqref{eq:reuse} quantifies how much leakage a forged certificate can
conceal as a function of the number of estimation rounds.
Corollary~\ref{cor:finite} shows the exponent is not merely asymptotic: it
upper-bounds the composable failure probability at every finite block
length, and Theorem~\ref{thm:composable} shows exactly where it lives inside
the composable finite-size architecture---the acceptance test is the
parameter-estimation step of the proof, and the certificate reusability rate
is the worst-case Holevo term of the key-length formula, with the Gaussian
confidence intervals of conventional finite-size
analysis~\cite{Leverrier2010} recovered as its near-threshold limit
(Corollary~\ref{cor:recovery}).

The strict monotonicity on which the invertibility of the rate rests is no
longer an assumption: Proposition~\ref{prop:mono} proves it by noise
injection, with the quantitative gap~\eqref{eq:monobound}, so
Theorem~\ref{thm:main}(iii), Corollary~\ref{cor:reuse}, and
Theorem~\ref{thm:composable} are unconditional in the trusted-correlation
model. Three problems remain open: an analytic proof of the concavity in
Proposition~\ref{prop:convex}(b), which would remove the remaining
assumption from Theorem~\ref{thm:convex} and is nontrivial because
$\chi_{BE}$ is there a difference of two concave functions; an analytic
proof of Condition~\ref{cond:reg} for general monitored sets
$\mathcal{M}$---for instance via the precision-coordinate convexity of
Remark~\ref{rem:condsupport}---which would make Theorem~\ref{thm:main}(iv)
unconditional; and the extension of Theorem~\ref{thm:composable} beyond the
trusted-correlation model, for which the Wishart statistics and residual
ingredients are identified in Remarks~\ref{rem:scope}
and~\ref{rem:outside}. The trusted-correlation results depend on none of
these.



\end{document}